\documentclass[journal]{IEEEtran}
\usepackage{graphicx}
\usepackage{color}
\usepackage{indentfirst}

\makeatletter
\let\theoremstyle\relax

\makeatother

\usepackage[none]{hyphenat}

\usepackage{amsmath}
\usepackage{amssymb}
\usepackage{amsfonts}
\usepackage{amsthm}

\usepackage{stfloats}
\usepackage{lipsum}

\renewcommand{\baselinestretch}{1.0}
\usepackage{cite}
\usepackage{multirow}
\usepackage{longtable}
\usepackage{rotating}
\usepackage{caption2}

\usepackage{makecell}

\usepackage{xcolor}
\usepackage{bigdelim}

\usepackage{setspace}
\theoremstyle{definition}
\newtheorem{theorem}{Theorem}
\newtheorem{proposition}{Proposition}

\newtheorem{corollary}{Corollary}
\newtheorem{definition}{Definition}

\newtheorem{example}{Example}
\newtheorem*{StateEstimation}{State Estimation Problem}
\newtheorem*{StateEstimationAttack}{State Estimation Problem under Attacks}

\usepackage{diagbox}
\usepackage{lineno}

\usepackage{algorithm}
\usepackage{algorithmic}

\begin{document}

\title{Error- and Tamper-Tolerant State Estimation for Discrete Event Systems {under Cost Constraints}}

\author{Yuting~Li,~Christoforos~N.~Hadjicostis,~\IEEEmembership{Fellow,~IEEE},~Naiqi~Wu,~\IEEEmembership{Fellow,~IEEE,}~and~Zhiwu~Li,~\IEEEmembership{Fellow,~IEEE}
\thanks{This work was supported in part by the National Key R\&D Program
	of China under Grant 2018YFB1700104, the National Natural
	Science Foundation of China under Grant 61873342, and the
	Science Technology Development Fund, MSAR, under Grant
	0012/2019/A1.}
\thanks{Y. Li and N. Wu are with the Institute of Systems Engineering, Macau University of Science and Technology, Taipa 999078, Macau SAR China (e-mail: yuutinglee@163.com; nqwu@must.edu.mo).}
\thanks{C. N. Hadjicostis is with the Department of Electrical and Computer
	Engineering, University of Cyprus, Nicosia 1678, Cyprus (e-mail: chadjic@ucy.ac.cy).}
\thanks{Z. Li is with the Institute of Systems Engineering, Macau University of Science and Technology, Taipa 999078, Macau SAR China, and  also with the School of Electro-Mechanical Engineering, Xidian University, Xi'an 710071, China (e-mail: zhwli@xidian.edu.cn).}}


\maketitle

\begin{abstract}
This paper deals with the state estimation problem {in} discrete-event systems modeled with nondeterministic finite automata, {partially observed via} a sensor measuring unit whose {measurements (reported observations)} may be {vitiated} by a malicious attacker.
{The} attacks {considered in this paper include} arbitrary deletions, insertions, or substitutions of {observed symbols by taking into account a bounded number of attacks} {or, more generally, a total cost constraint (assuming that each deletion, insertion, or substitution bears a positive cost to the attacker)}. 
An efficient approach is proposed to describe possible  sequences of observations that match the {one {received} by} the measuring unit, as well as their corresponding state estimates and associated total costs.
We develop an algorithm to {obtain} the least-cost {matching} sequences by reconstructing only a finite number of possible sequences, {which we subsequently use to efficiently} {perform} state estimation.
We also develop a {technique} for {verifying} {tamper-tolerant} diagnosability under {attacks} that involve a bounded number of deletions, insertions, and substitutions (or, more generally, {{under} attacks} of bounded total cost) by using a {novel structure} {{obtained} by attaching attacks and costs to the original {plant}}. {The overall {construction and verification procedure have} complexity that is of $O(|X|^2 C^2)$, where $|X|$ is the number of states of the given finite automaton and $C$ is the maximum total {cost} that is allowed for all the deletions, insertions, and substitutions.}
{We determine the minimum value of $C$ such that the attacker can coordinate its tampering action to keep the observer indefinitely confused while utilizing a finite number of attacks. Several examples are presented to demonstrate the proposed methods.}
\end{abstract}


\IEEEpeerreviewmaketitle

\section{Introduction}
\IEEEPARstart{S}{tate} estimation {in} continuous-time systems {was} initiated in the 1950s {and} has been {extensively} applied to {a variety of} areas of engineering and science 
\cite{simon2006optimal}. The {primary} motivation for state estimation is {to be able to perform} analysis of the current state of a system under the conditions characterized by a streaming sequence of measurements.
{The state estimator has knowledge of both} a model of the system and the way it generates observations (outputs). Under appropriate redundancy levels, it can eliminate the effects of bad or erroneous measurements ({in some cases,} even account for temporary loss of measurements) without significantly affecting the quality of {estimated} values \cite{monticelli2012state}.

The development of information and computer technology {has spurred} the booming of computer-integrated systems whose {structure and} evolution are regulated by engineers; {examples include} manufacturing systems, intelligent traffic systems, and communication networks. Discrete-event systems (DESs) are a technical abstraction of these systems with discrete state spaces and event-triggered dynamics \cite{zeigler2000theory}.
The state estimation problem {in} {DESs} is essential {since typically} state information cannot be {directly} obtained due to limited sensor availability {in many applications of DESs}. For example, state estimation is {critical} for supervisory control \cite{ramadge1987supervisory,hadjicostis2020estimation}, fault diagnosis \cite{sampath1995diagnosability,debouk2000coordinated,debouk2003effect}, and opacity {verification and enforcement} \cite{bryans2008opacity,saboori2007notions,jacob2016overview,saboori2013verification,saboori2008opacity}.
The problem becomes challenging because of {possible} faulty observations (e.g., due to cyber attacks, malfunctioning sensors, erroneous communication transmissions, or synchronization issues during the transmission of information from different sensors) \cite{wakaiki2017supervisory,hu2018state,athanasopoulou2010maximum,debouk2003effect,carvalho2011generalized,carvalho2012robust,carvalho2013robust}.

In cyber-physical systems, it is common to encounter situations, where serious risks of cyber attacks {occur} between cyber and physical components.
Cyber attacks can lead to enormous {financial} {loss} and disorder of {important socio-economical infrastructures} \cite{mousavinejad2018novel, ding2018survey,li2013robust,zhao2015power}. {Examples of cyber attacks include the} StuxNet {strike on}  industrial control systems \cite{farwell2011stuxnet}, {{the} hacking of the} Maroochy Shire Council's sewage control system (resulting in the release of one
million liters of untreated sewage) \cite{slay2007lessons}, {and {the} spoofing of global positioning systems} {to capture} unmanned aircrafts \cite{kerns2014unmanned}. 
{This paper addresses} centralized state estimation and fault diagnosis {in DESs}  under adversarial {attacks}  {that corrupt the sensor readings}. Some related work has appeared in the context of sensor attacks that drive  a controlled DES to unsafe or undesirable states by manipulating observation sequences \cite{goes2017stealthy,meira2019synthesis,su2017cyber,su2018supervisor}.

The work in this paper is also related to some existing state estimation and security {results} in the area of DESs \cite{athanasopoulou2010maximum, wakaiki2017supervisory,thorsley2008diagnosability,athanasopoulou2006probabilistic,boel2002decentralized,khanna1973sampling,lin2014control,lin2019state}.  In particular, {the study in \cite{athanasopoulou2010maximum}} considers fault diagnosis under unreliable observations: transpositions, deletions, and insertions of output symbols are formally defined with probabilities {captured by a probabilistic finite automaton}.
Drawing upon a probabilistic methodology, the work in  \cite{athanasopoulou2010maximum} determines  {{whether} the} fault-free {or} faulty {system} has {most likely} generated the sequence {received} at a diagnoser. In \cite{wakaiki2017supervisory}, a supervisor of a plant {{under partial observations}} is constructed to {overcome} attacks, {where  attacks are modeled} by a set-valued map that represents all possibly corrupted strings with respect to each original string.

{The authors of \cite{boel2002decentralized} consider decentralized fault diagnosis, where communication between two diagnosers is {expensive}. The costs
	on the communication channels are described in terms of {the number of data packets}.
	One diagnoser aids {the other} in achieving failure detection and diagnosis by sending information {about} its estimated states. 
	In order to {perform} decentralized fault diagnosis and minimize the costs of communication and computation, a protocol is implemented to decide {what} {kind} of information is useful to communicate between the diagnosers. The work in \cite{rudie2003minimal} addresses the problem of decentralized state estimation with costly communication between two agents (or local sites). In order to minimize communication {costs},  a communication strategy, {i.e.,} a set of {functions}, is developed to determine whether {a state estimated by one} agent should be communicated to {the other}. 
}

With the development of networked control systems, {data exchanged} among networked components may suffer {communication errors or malicious} {attacks}  \cite{hu2018state,ding2018survey,mousavinejad2018novel}. In the framework of DESs,  three typical types of cyber attacks {are considered, namely} deletions, insertions, and substitutions. 
{A} deletion (substitution) attack is {a} natural strategy, in which a valid data transmission is maliciously deleted (substituted) such that {a} system {may} deviate from {its} expected {behavior} {(if this symbol is a control input) or an outside observer may incorrectly estimate its activity (if this symbol is an output of the system)} \cite{wakaiki2017supervisory}.
{An} insertion attack {is an attack that inserts extraneous symbols and can have similar effects as above. An insertion attack can also be used to render} {certain} resources of a system unavailable, e.g., an attacker sends a huge number of {fabricated} packets to a device{,} with the intention of dramatically consuming 
amounts of endpoint network bandwidth \cite{householder2001managing}. 

%
%


{In order to overcome such corruptions, the strategy proposed in this paper ensures that} the sequence estimation unit calculates a set of matching sequences {based on the possibly tampered sequence} received from the channel.
We aim to choose {among} all matching sequences the ones that closely match the one {received} at the sequence estimation unit.
{Therefore, a} cost-value {notion} is proposed, where a positive value is assigned to each type of attack; {this value} is inversely related to {{the} {likelihoods} of {different types of} attack occurrences}. 
Matching sequences with less costs are much more {likely to have occurred and can be used to make educated estimates of states ({or} faults) that may have occurred in the system}.
{Note that} the number of matching sequences and their lengths can be infinite due to the existence of deletions.
Hence, in this paper, an upper bound {on the} total cost is set to limit the number of sequences that match an observed sequence. A special case of this setting is to assume that the total number of attacks is bounded (this case arises when each attack has a bounded cost).

The main contributions of this paper are as follows:

1) It formulates and solves the state estimation problem under {communication} attacks of bounded total cost, where each type of attack is associated with an individual positive cost.

2) It proposes an efficient state estimation algorithm by representing all matching sequences {as} {the} language of an observation automaton that synchronizes with the plant.

3) A {novel structure} is {proposed} to check {the} {tamper-tolerant} diagnosability of {the} plant by attaching attacks and costs to an enhanced version on the plant model.


\section{Background and Preliminaries}\label{Notion}
Let $\Sigma$ be an alphabet with a set of distinct symbols (events) $\alpha, \beta, ...$. {As usual,} $\Sigma^{\ast}$ denotes the set of all finite symbol sequences over $\Sigma$, including the empty sequence $\varepsilon$ (sequence with no symbols).
A member of $\Sigma^{\ast}$ is {said to be} a string or trace, and a subset of $\Sigma^{\ast}$ is a language {defined over $\Sigma$}.
The length of a string {$s\in\Sigma^{\ast}$} is {the number of symbols in $s$,} denoted by $|s|$ with $|\varepsilon|=0$.
Given strings $s$, $t\in \Sigma^{\ast}$, the concatenation of strings $s$ and $t$ is defined as the string $st$.
For a string $s\in \Sigma^{\ast}$, $t \in \Sigma^{\ast}$ is said to be a \emph{prefix} of $s$, if $(\exists t^{\prime}\in \Sigma^{\ast})$ $s=tt^{\prime}$.
Given a language $\mathcal{L}\subseteq \Sigma^{\ast}$,
$\bar{\mathcal{L}}$ denotes the \emph{prefix-closure} of $\mathcal{L}$, defined as
$\bar{\mathcal{L}}=\{t\in \Sigma^{\ast}|\exists t^{\prime}\in \Sigma^{\ast},tt^{\prime}\in \mathcal{L}\}$.
{By a slight abuse of notation}, for $\sigma\in \Sigma$ and $s\in\Sigma^{*}$, we {write} $\sigma \in s$ to represent {that} {the} event  $\sigma$ {is} in $s${, i.e., $s=s'\sigma s''$ for some $s',s''\in \Sigma^{\ast}$.} For $\Sigma'\subseteq \Sigma$, we write $\Sigma' \in s$ to denote ($\exists \sigma\in \Sigma'$) $\sigma\in s$; otherwise $\Sigma' \notin s$. We {use} $\mathcal{L}/s$ {to denote} the \emph{postlanguage} of $\mathcal{L}$ after $s$, i.e., $\mathcal{L}/s = \{t\in \Sigma^\ast| st\in \mathcal{L}\}$.
\begin{definition}
	A deterministic finite automaton (DFA), denoted by $G$, is a {four}-tuple $G =(X, \Sigma, \delta, x_0)$, where $X$ is the set of states, $\Sigma$ is the set of events, $\delta:X\times \Sigma \rightarrow X$ is the {partial} state transition function, and $x_0 \in X$ is the initial state.
	
	For convenience, $\delta$ can be extended from domain $X\times \Sigma$ to
	$X\times \Sigma^{\ast}$ in the following recursive manner:
	$\delta(x,\varepsilon)= x$; 
	$\delta(x, \sigma s)=\delta(\delta(x, \sigma),s)$ for $x\in X$, $\sigma \in \Sigma$, and $s\in \Sigma^{\ast}$ if $\delta(x,\sigma)$ is defined. Note that if $\delta(x,\sigma)$ is not defined, then $\delta(x, \sigma s)$ is not defined.
	The generated {language} of $G$ is given by {$\mathcal{L}(G)=\{s\in \Sigma^{\ast}|\delta(x_0, s)!\}$}, where $!$ means ``is defined''.
\end{definition}

\begin{definition}		
	A nondeterministic finite automaton (NFA), denoted by $G_{nd}$, is a four-tuple $G_{nd}=(X, \Sigma, \delta, X_0)$, where $X$ and $\Sigma$ have the same interpretation as {in a} DFA, $\delta:X\times \Sigma \rightarrow 2^{X}$ is the {(nondeterministic)} state transition {function}, and  {$X_{0}\subseteq X$ is a set of initial states}.
	
	{By} letting $B\subseteq X$ and $\sigma \in \Sigma$, $\delta(B, \sigma)$ is defined as $\cup_{x\in B}\delta(x, \sigma)$.
	In order to characterize the strings generated by {an} NFA, the domain $X\times \Sigma$ {of the transition {function} can} be extended to $X\times \Sigma^{\ast}$. For $x\in X$, $s\in \Sigma^{\ast}$, and $\sigma\in \Sigma$, $\delta$ is defined recursively as:  {$\delta(x,\varepsilon)= \{x\}$; $\delta(x, \sigma s)=\delta(\delta(x,\sigma),s):=\cup_{x'\in \delta(x,\sigma)}\delta(x',s)$}. An event $\sigma \in \Sigma$	is said to be \emph{feasible} at state $x\in X$ if  $\delta(x, \sigma)$ is {non-empty}. The {language} generated by $G_{nd}$ {is} defined as $\mathcal{L}(G_{nd})=\{s\in \Sigma^{\ast}|\exists x \in X_0,\delta(x, s)\neq \emptyset\}$, where $\emptyset$ denotes the empty set. {The} language $\mathcal{L}(G_{nd})$ is said to be \emph{live} if whenever $s \in \mathcal{L}(G_{nd})$, there exists an event $e\in \Sigma$ such that $se \in \mathcal{L}(G_{nd})$ \cite{sampath1995diagnosability}.
\end{definition}

{The set of events} $\Sigma$ {in a DFA or NFA} is partitioned into the subset of observable events, $\Sigma_{o}$, and the subset of unobservable
events, {$\Sigma_{uo}$ with} $\Sigma_{uo}=\Sigma \setminus \Sigma_{o}$.
{The sensor measuring unit can only observe and record observable events}.
{The natural projection $P:\Sigma^{\ast}\rightarrow \Sigma_{o}^{\ast}$ captures the sequence of observable {actions} in response to a sequence of events $s\in \mathcal{L}(G_{nd})$;} it is defined recursively as
\[
P(\sigma)=\left\{
\begin{aligned}
\sigma~ &~ \text{if} ~\sigma \in \Sigma_{o}, \\
\varepsilon~ &~ \text{if} ~ \sigma \in \Sigma_{uo}\cup \{\varepsilon\},
\end{aligned}
\right. 
\]
and $P( s\sigma)=P(s)P(\sigma),\text{for}~ \sigma \in \Sigma, s \in \Sigma^{\ast}$.
{The} natural projection $P$ can be used to map any trace $s\in \Sigma^{\ast}$ {to the corresponding sequence of observations $P(s)$ observed} at the sensor measuring unit. The inverse projection of $P$, $P^{-1}:\Sigma_{o}^{\ast}\rightarrow\Sigma^{\ast}$, is defined as follows: for all $\omega \in \Sigma_{o}^{\ast}$
\[P^{-1}(\omega)=\{s\in \Sigma^*|P(s)=\omega\}.\]
{A typical task by an} observer/agent {is} to determine a set of possible states {in}  which a system  may be. {The state} estimation problem {in} DESs is defined as follows.

\begin{StateEstimation}
	Given a {DES} described by NFA $G_{nd}$ {with a} sensor measuring unit, an observer/agent {needs to determine} a set of possible states {based on} an observation sequence {$P(s)\in \Sigma^{\ast}_{o}$} (generated by an underlying sequence of events $s,s\in \mathcal{L}(G_{nd})$, in the given NFA) {that is received} from the sensor measuring unit.
	The set of all possible states corresponding to an observable sequence $\omega=P(s) \in \Sigma_{o}^{\ast}$ starting from {the} states in {a} set $B$ {with} $B\subseteq X$ {is defined as}
	$R(B, \omega)=\{x^{\prime}\in X|(\exists s\in \Sigma^{\ast}) (\exists x \in B) \{P(s)=\omega\wedge x^{\prime}\in\delta(x,s)\}\}$.
\end{StateEstimation}

\begin{definition}\label{defObs}
	
	An observer is captured by {$Obs(G_{nd})=AC(2^X, \Sigma_{o}, \delta_{obs}, R(X_0, \varepsilon)):=(X_{obs}, \Sigma_{o}, \delta_{obs}, x_{0,obs})$}, where $2^{X} $ is the set of distinct subsets of $X$ (i.e., {the powerset of the set of states} of the given NFA $G_{nd}=(X, \Sigma, \delta, X_0)$), $\Sigma_{o}$ is the set of observable events, $x_{0,obs} \in 2^X$ is the set of initial states {given by $x_{0,obs}=R(X_0, \varepsilon)$}, and $\delta_{obs}: 2^X \times \Sigma_{o} \rightarrow 2^X$ is the state transition function {defined for $B\in 2^X$ and $\sigma_{o}\in \Sigma_{o}$ as $\delta_{obs}(B,\sigma_{o})=R(B,\sigma_{o})$. $AC(\cdot)$ denotes the accessible part of the observer starting from $x_{0,obs}$}.	
	
	For the construction of $\delta_{obs}$ over the domain $X_{obs} \times \Sigma_{o}^{\ast}$, one can proceed recursively as follows.
	First, for $x_{obs} \in X_{obs}$, we set $\delta_{obs}(x_{obs}, \varepsilon) = R(x_{obs}, \varepsilon)$.
	Second, for $\omega\in \Sigma_{o}^{\ast}$, $\sigma_{o} \in \Sigma_{o}$, we set $\delta_{obs}(x_{obs}, \sigma_{o} \omega) = \delta_{obs}(\delta_{obs}(x_{obs}, \sigma_{o}), \omega)= \cup_{x'\in \delta_{obs}(x_{obs}, \sigma_{o})} \delta_{obs}(\{x'\}, \omega)$.
\end{definition}


\begin{example}\label{observer of NFA}
	{Consider the} NFA $G_{nd}$ shown in Fig.~\ref{NFA}, where $X=\{0,1,2,3,4\}$, $\Sigma=\{\alpha, \beta,\gamma, \zeta\}$, $\Sigma_{o}=\{\alpha, \beta,\gamma\}$, $\Sigma_{uo}=\{ \zeta\}$, {$\delta$ is as defined in the figure}, and $X_{0}=\{0,1,2,3,4\}$.
	Note that, we have $\delta(\{2\},\alpha\beta\alpha)=\emptyset$ and $\delta(\{2\}, \beta\alpha\alpha)=\delta(\delta(\delta(\{2\}, \beta), \alpha), \alpha)=\{3,4\}$.
	
	Initially, the set of possible states is $x_{0,obs}=X_0$.
	{For} $s=\alpha\beta\alpha$, we can infer the following sets of state estimation:  \[\{0,1,2,3,4\}\xrightarrow{\alpha}\{2,3,4\}\xrightarrow{\beta}\{2,3\}\xrightarrow{\alpha}\{3,4\}.\]
	Note that this is also reflected in the observer in Fig.~\ref{ObserverNFA}. We start in state $x_{0,obs}$ (marked by an arrow); {if} $\alpha$ is observed, we {reach} state \{2,3,4\}; {if $\beta$ is subsequently observed, we {reach} $\{2,3\}$ from $\{2,3,4\}$; and so forth.}
	
	%
	%
	%
	
	\begin{figure}[htb]
		\begin{minipage}[l]{1\linewidth}
			\centering
			\includegraphics[scale=0.35]{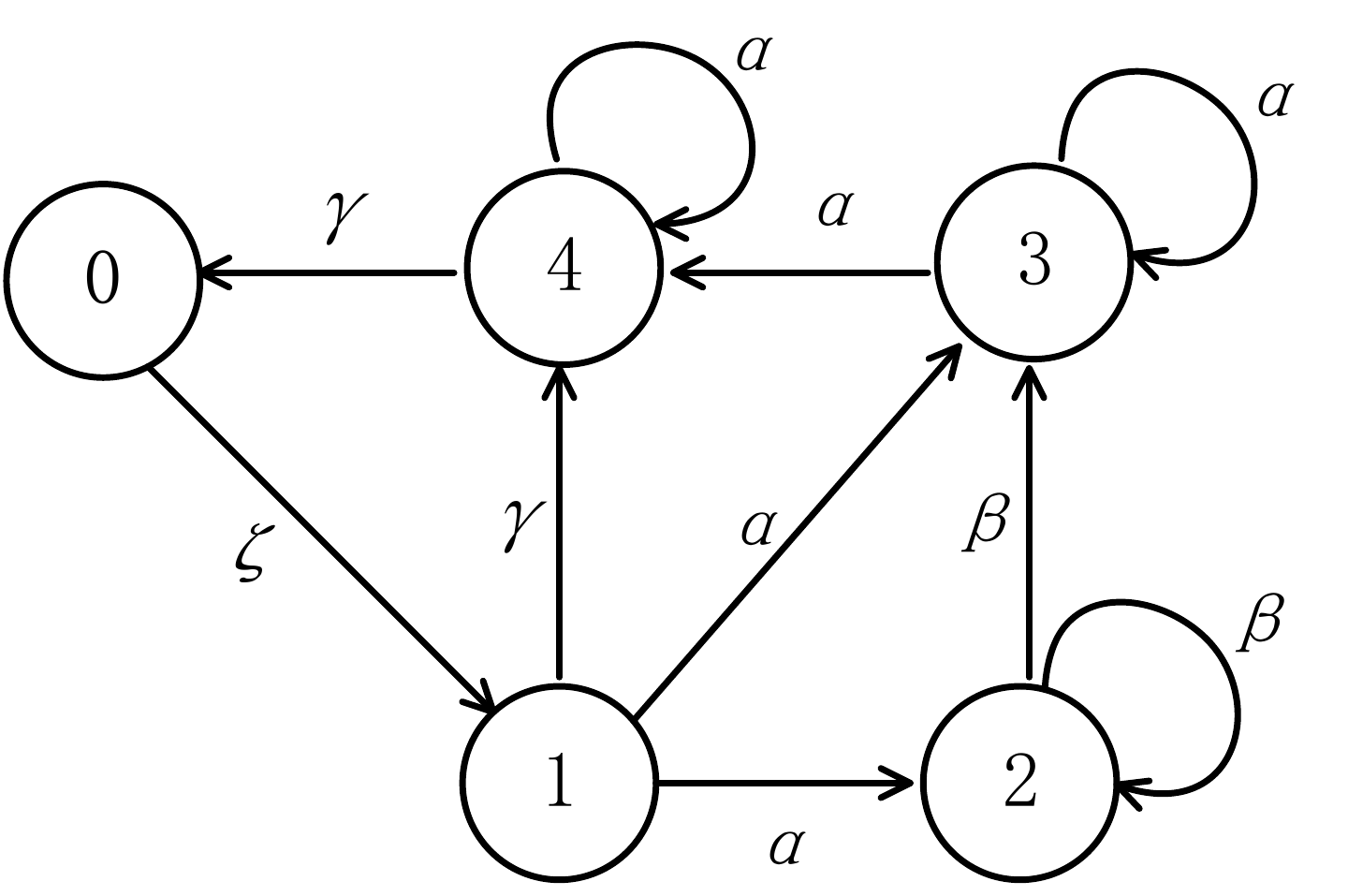}
			\caption{Nondeterministic finite automaton.}
			\label{NFA}
		\end{minipage}
	\end{figure}
	
	\begin{figure}[htb]	
		\begin{minipage}[l]{1\linewidth}
			\centering
			\includegraphics[scale=0.35]{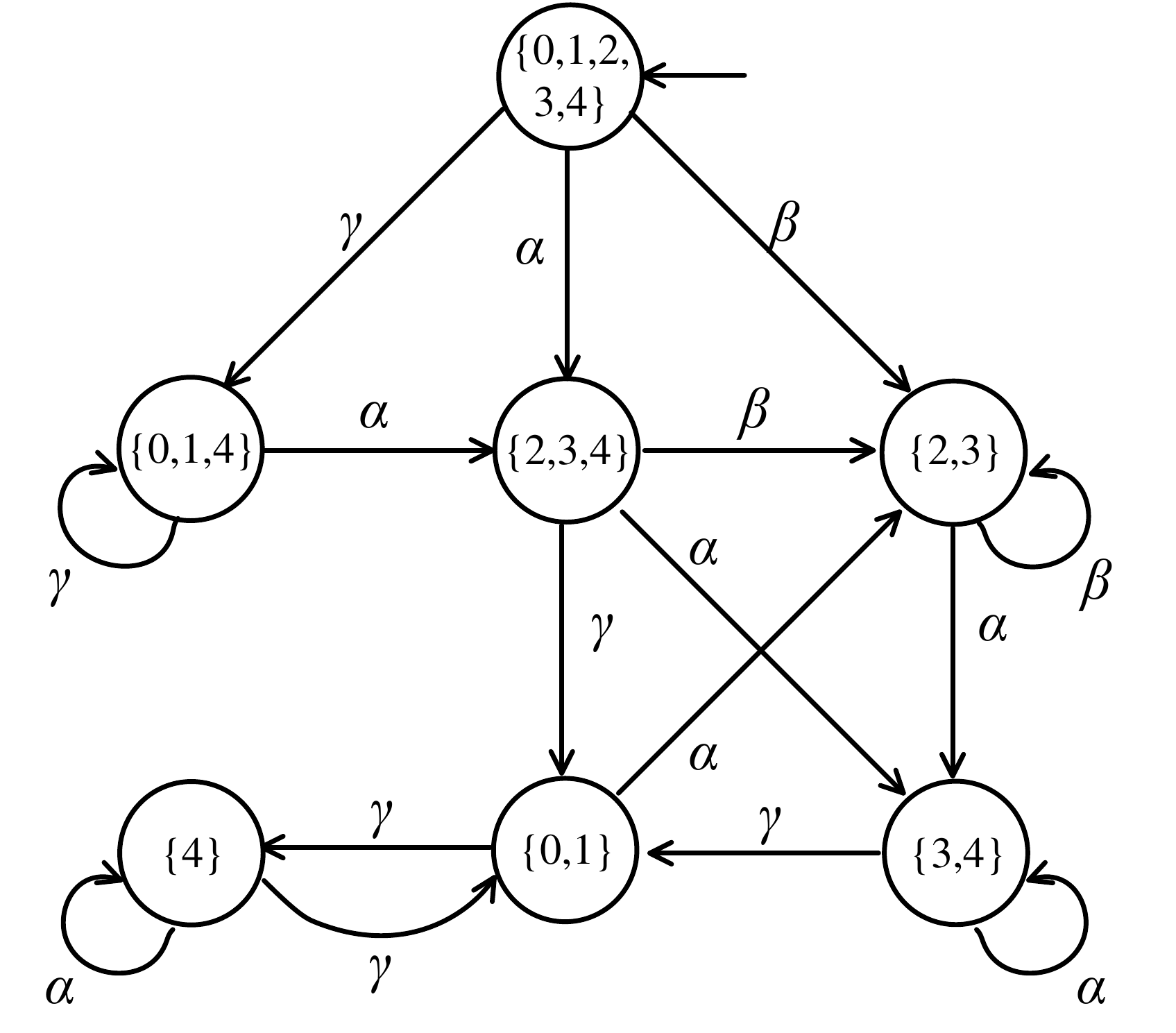}
			
			\caption{{Observer for NFA in} Fig.~\ref{NFA}.}\rm
			\label{ObserverNFA}
		\end{minipage}
		
	\end{figure}
	
\end{example}				

\section{Observation Sequences under Attacks} \label{TransAttack}

In general, {malicious} attacks may corrupt  sequences at the communication channel, {{such} that the} sequence {received} at the sequence estimation unit is unreliable.
In this section, we {propose} a compact way to represent {possibly} matching sequences and {describe} an efficient method to reduce {the number of such sequences that need to be explored}.
In {{the}} next section, {{we devise {another} way to {filter}, among the matching sequences, the sequences that}}  belong  to  the behavior that can be generated by the NFA, {and} {subsequently} use them to {perform} state estimation according to their {costs}. 

Referring to Fig.~\ref{InternalStructure}, if the plant generates {a}
string $s\in \mathcal{L}(G_{nd})$, the observed string at the sensor measuring unit is $\omega=P(s)$.
An attacker may corrupt the output signals produced by the sensor measuring unit by deleting, inserting, or substituting certain types of events. {The resulting tampered observation sequence} is denoted
as $\omega_{A}\in A(\omega)$, where $A(\omega)$ is a set of tampered sequences {that can be} generated by the attacker. Based on $\omega_{A}$, the sequence estimation unit calculates a set of matching sequences $RA(\omega_{A})$ {that} is used to perform state estimation. 

\begin{figure}[htb]
	\begin{minipage}[l]{1\linewidth}
		\centering
		\includegraphics[scale=0.31]{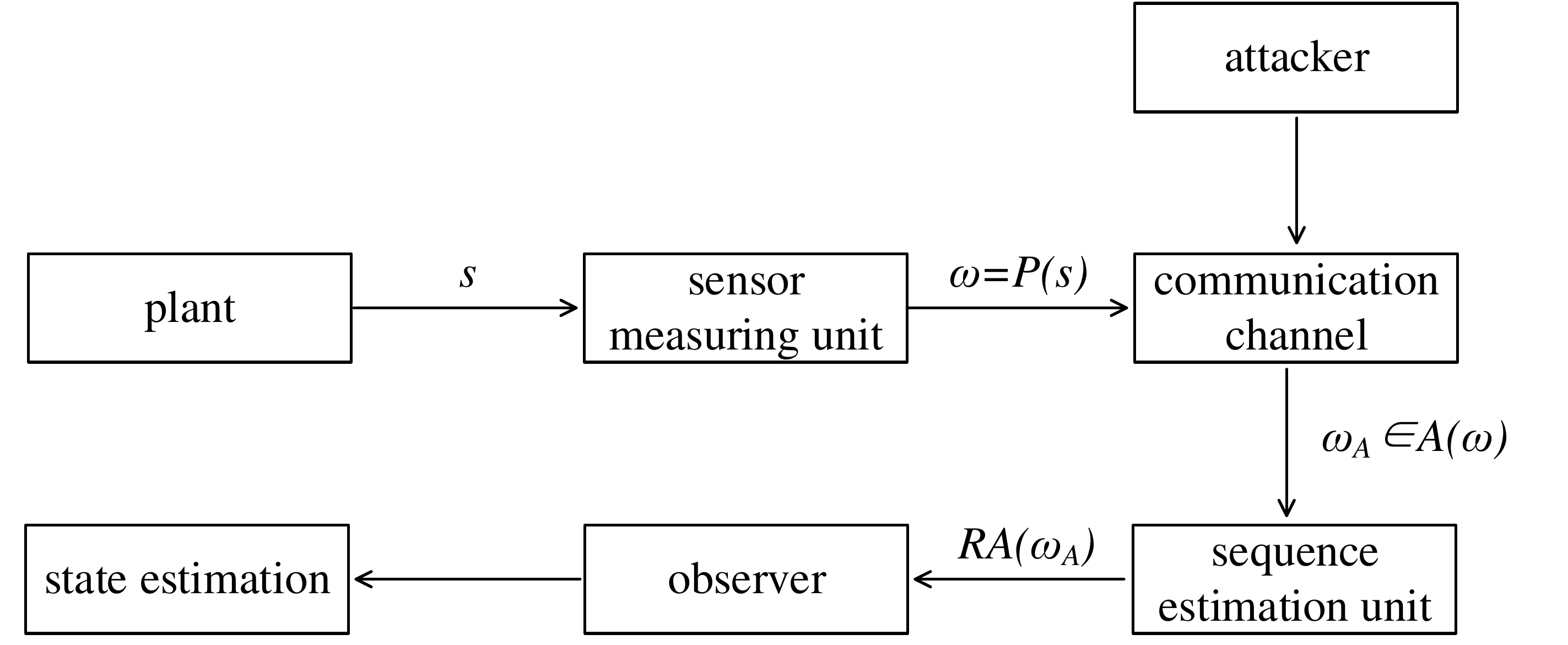}
		
		\caption{Attack setting.}
		\label{InternalStructure}
	\end{minipage}
	
\end{figure}

We {focus} on attacks due to symbol deletions, insertions, and substitutions.
In order to have a general form of attacks, suppose that each event $\sigma_{o} \in \Sigma_{o}$ {can be} associated with {some} arbitrary replacements ({for example, given} $\Sigma_{o}=\{\alpha, \beta, \gamma\}$, $\alpha$ can be replaced by $\beta$ or $\gamma$), {and} some events {may also} be deleted or inserted under attacks.
Note that {we assume that each symbol (in a sequence of symbols)} received at the sequence estimation unit {can only be} related to at most one type of attack. {In other words}, it is not {possible} for {the} attacker to corrupt the same observable event more than once.

{More specifically,} the attacker has the capability to
\begin{enumerate}
	\item delete certain types of events from a set $\Sigma_{D}\subseteq \Sigma_{o}$;
	\item insert certain types of events from a set $\Sigma_{I}\subseteq \Sigma_{o}$;
	\item substitute an event $\sigma_{oi}\in \Sigma_{o}$ with an event $\sigma_{oj}\in \Sigma_{o}$ for some pairs $(\sigma_{oi},\sigma_{oj})$ of events captured in the set $\Sigma_{T}\subseteq (\Sigma_{o}\times\Sigma_{o})\setminus\{(\sigma_{oi}, \sigma_{oi})|\sigma_{oi}\in \Sigma_{o}\}$.
\end{enumerate}

Suppose that each individual deletion, insertion, or substitution of an event is associated with a positive cost.		
Costs {capture in some sense} the {expense} of {the attacker when trying to {alter}} symbols of transmitted sequences at the communication channel. The type of {attack} with costs can be summarized by a table, as illustrated in the following example.

\begin{example}\label{attacked sequences at attacker}
	{Let us consider} the system in Fig.~\ref{NFA}. 
	Suppose that $\Sigma_{I}=\{\beta\}$, $\Sigma_{D}=\{\alpha\}$, and $\Sigma_{T}=\{(\alpha,\beta),(\gamma,\alpha)\}$. When $\alpha$ is corrupted to $\beta$, the attacker spends two units. Similarly, when $\gamma$ is corrupted to $\alpha$, it spends one unit. The cost of one-step deletion of $\alpha$ is three {units} and one-step insertion of $\beta$ is two {units}.

	\begin{table}[htb]
		\caption{Attacks with costs}\label{attacksCost}
		\begin{tabular} {|c| p{0.87cm}<{\centering}|p{0.87cm}<{\centering}| p{0.87cm}<{\centering}|p{0.87cm}<{\centering}|}
			\hline
			\diagbox{original}{attack}& $\alpha$& $\beta$ & $\gamma$ & $\varepsilon$ \\
			
			\hline
			$\alpha$& & 2& &3\\
			\hline
			$\beta$& & & &\\
			\hline
			$\gamma$&1&& &\\
			\hline
			$\varepsilon$& &2&& \\
			\hline
			
		\end{tabular}
	\end{table}
\end{example}

In Table~\ref{attacksCost}, Column 1 represents {the {symbol} originally generated by the system} (including the empty {symbol}) and Row 1 shows {possible corruptions} due to attacks. Note that $\varepsilon$ in Column 5 means that an original event can be deleted, {whereas $\varepsilon$} in Row 5 stands for insertions of events.

Given a sequence of {observations} $\omega\in \Sigma_{o}^*$, we can {systematically} obtain a set of possibly tampered sequences {that can be generated by} the attacker. Suppose that $\omega=\sigma_{o1}\sigma_{o2}...\sigma_{on}$, where $\sigma_{oi}\in\Sigma_{o}$ and $i\in \{1,2,...,n\}$. The set of possibly tampered sequences, denoted by $A(\omega)$, is defined as $A(\omega)=\Sigma_{I}^*(\sigma_{o1}+\sigma_{o1}')\Sigma_{I}^*(\sigma_{o2}+\sigma_{o2}')\Sigma_{I}^*...\Sigma_{I}^*(\sigma_{on}+\sigma_{on}')\Sigma_{I}^*$, where $\sigma_{oi}'=\sigma_{oi,D}+\sigma_{oi,T}$ ($i\in \{1,2,...,n\}$) with

$\sigma_{oi,D}=\left\{
\begin{aligned}
\varepsilon~~~&~ \text{if} ~\sigma_{oi}\in \Sigma_{D}, \\
\sigma_{oi}~&~ \text{if} ~\sigma_{oi}\not\in \Sigma_{D}, \\
\end{aligned}
\right.$\\

$\sigma_{oi,T}=\left\{
\begin{aligned}
\sigma_{oj_1}+\sigma_{oj_2}+...+\sigma_{oj_k}~&~ \text{if}~  \{j_1,j_2,...,j_k\}\\&=\{j|(\sigma_{oi},\sigma_{oj})\in \Sigma_{T}\},\\
\\
\sigma_{oi}~~~~~~~~~~~~~~~~~~~~~~~~~&~ \text{if}~\{j|(\sigma_{oi},\sigma_{oj})\in \Sigma_{T}\}\\&=\emptyset. \\
\end{aligned}
\right.$

Note that in the above expression we adopted the symbol ``$+$'' to represent the logical ``OR'' function.
An upper bound {on the total cost (i.e., the sum of costs over all tampered symbols in the sequence)}, denoted by $C$, is enforced to limit the number of possibly tampered sequences. We use $A_{C}(\omega)$ to restrict $A(\omega)$ to a set of {pairs involving a string from $A(\omega)$ and its associated total cost, where the} maximum {total} cost is $C$. 
Note that it is possible that the same string can be generated by the attacker with different total costs. In this case, we associate with {the string the smallest cost}.
\begin{example}\label{corrupted sequences with cost}
	{Consider again} the system in Fig.~\ref{NFA}.
	Suppose that $\zeta\alpha\alpha\alpha$ is generated by the plant. The attacker can observe $\omega=\alpha\alpha\alpha$ and {may} corrupt this sequence {using any of the} type of {attacks} shown in Table~\ref{attacksCost}. {Therefore, in this case,}
	$A(\omega)=\beta^*(\alpha+\varepsilon+\beta)\beta^*(\alpha+\varepsilon+\beta)\beta^*(\alpha+\varepsilon+\beta)\beta^*=
	\{\alpha\alpha\alpha$,
	$\beta\alpha\alpha\alpha$,
	$\alpha\beta\alpha\alpha$,
	$\alpha\alpha\beta\alpha$,
	$\alpha\alpha\alpha\beta$, 
	$\beta\alpha\alpha$,
	$\alpha\beta\alpha$,
	$\alpha\alpha\beta$,
	$\alpha\alpha$,
	$\beta\beta\alpha$,
	$\beta\alpha\beta$, 
	$\alpha\beta\beta$,
	$\beta\beta\alpha\alpha\alpha, ...\}$. {If we} set {the} upper bound {on the total cost to two}, {we {obtain}} $A_2(\omega)=\{(\alpha\alpha\alpha,0)$,
	$(\beta\alpha\alpha\alpha,2)$, 
	$(\alpha\beta\alpha\alpha,2)$,
	$(\alpha\alpha\beta\alpha,2)$,
	$(\alpha\alpha\alpha\beta,2)$, 
	$(\beta\alpha\alpha,2)$, 
	$(\alpha\beta\alpha,2)$,     
	$(\alpha\alpha\beta,2)\}$.
\end{example}

For {clearer notation}, we  define the set of {deleted} label{{s}}
$D=\{d_{\sigma_{oi}}|\sigma_{oi}\in \Sigma_{D}\}$, where $d_{\sigma_{oi}}$ denotes the deletion of $\sigma_{oi}$; the set of {inserted} label{{s}} $I=\{i_{\sigma_{oj}}|\sigma_{oj}\in \Sigma_{I}\}$, where $i_{\sigma_{oj}}$ denotes the insertion of $\sigma_{oj}$; and the set of {attacked} label{{s}} 
$T=\{t_{\sigma_{oi}\sigma_{oj}}|(\sigma_{oi}, \sigma_{oj})\in \Sigma_{T}\}$, where $t_{\sigma_{oi}\sigma_{oj}}$ denotes the substitution of $\sigma_{oi}$ by $\sigma_{oj}$. {The above attack forms} are captured by the set of {attacked} labels $AT= D\cup I\cup T$.
For example, {suppose} that $\beta$ can be inserted and $\gamma$ can be replaced by $\beta$ at a communication channel under attacks. If $\beta$ is received at the sequence estimation unit,  possible original sequences could be $\varepsilon$, $\beta$, or $\gamma$. {In order to clarify the type of attack, the sequences $\varepsilon$ and $\gamma$ are relabeled respectively by {$i_\beta$} and $t_{\gamma\beta}$}.

{At} the sequence estimation unit, given a possibly tampered sequence $\omega_{A}\in A(\omega)$, we can obtain the set of all matching sequences, denoted by $RA(\omega_{A})$. Suppose that $\omega_{A}=\sigma_{A1}\sigma_{A2}...\sigma_{Am}$, where $\sigma_{Ai}\in\Sigma_{o}$ and $i\in \{1,2,...,m\}$. The set of all matching sequences at the sequence estimation unit is defined as $RA(\omega_{A})=D^*(\sigma_{A1}+\sigma_{A1}')D^*(\sigma_{A2}+\sigma_{A2}')D^*...D^*(\sigma_{Am}+\sigma_{Am}')D^*$, where $\sigma_{Ai}'=\sigma_{Ai,I}+\sigma_{Ai,T}$ {($i\in \{1,2,...,m\}$)} with

$\sigma_{Ai,I
}=\left\{
\begin{aligned}
i_{\sigma_{Ai}}~&~ \text{if} ~\sigma_{Ai}\in \Sigma_{I}, \\
\sigma_{Ai}~&~ \text{if} ~\sigma_{Ai}\not\in \Sigma_{I}, \\
\end{aligned}
\right.$\\

$\sigma_{Ai,T}=\left\{
\begin{aligned}
t_{\sigma_{Aj_1}\sigma_{Ai}}+t_{\sigma_{Aj_2}\sigma_{Ai}}+...+t_{\sigma_{Aj_{k'}}\sigma_{Ai}}~~~~~~~~~~~~~\\~~~~~ \text{if}~  \{j_1,j_2,...,j_{k'}\}=\{j|(\sigma_{Aj},\sigma_{Ai})\in \Sigma_{T}\},\\
\\
\sigma_{Ai}~~~~~~~~~~~~~~~~~~ \text{if} ~\{j|(\sigma_{Aj},\sigma_{Ai})\in \Sigma_{T}\}=\emptyset.
\end{aligned}
\right.$

Similarly, each sequence $\omega_{R}\in RA(\omega_{A})$ can be augmented with a cost value. 
Let $c_{t_{\sigma_{oi}\sigma_{oj}}}$, $c_{d_{\sigma_{oi}}}$, and $c_{i_{\sigma_{oj}}}$ respectively denote the costs of recovering {one-step} substitution of $\sigma_{oj}$ for $\sigma_{oi}$, deletion of  event $\sigma_{oi}\in \Sigma_{D}$, and insertion of event $\sigma_{oj}\in \Sigma_{I}$, where $c_{t_{\sigma_{oi}\sigma_{oj}}},c_{d_{\sigma_{oi}}},c_{i_{\sigma_{oj}}} > 0$.

We introduce a {cost function} $\Pi_{c}: (\Sigma_{o}\cup AT)^* \rightarrow \mathbb{N}$ from a matching sequence to its cost, where $\mathbb{N}=\{0,1,2,3,...\}$. 
More specifically, $\Pi_{c}$ is used to accumulate {the} total cost of attacks occurred at each matching {sequence}. The cost function $\Pi_{c}$ can be defined recursively as:
\[
\Pi_{c}(\sigma_{R})=
\begin{cases}
0& \text{if}~\sigma_{R} \in \Sigma_{o}\cup\{\varepsilon\},\\
c_{d_{\sigma_{oi}}}& \text{if}~\sigma_{R}=d_{\sigma_{oi}} \in D, \\
c_{i_{\sigma_{oj}}}& \text{if}~\sigma_{R}=i_{\sigma_{oj}} \in I,\\
c_{t_{\sigma_{oi}\sigma_{oj}}}~ &\text{if} ~\sigma_{R}=t_{\sigma_{oi}\sigma_{oj}} \in T,\\
\end{cases}
\]
and $\Pi_{c}(\omega_{R}\sigma_{R})=\Pi_{c}(\omega_{R})+\Pi_{c}(\sigma_{R})$, for $\omega_{R}\in(\Sigma_{o}\cup AT)^*$, $\sigma_{R}\in \Sigma_{o}\cup AT$. {If we} set the same upper bound {on the total cost to} $C$ among {each sequence} in $RA(\omega_{A})$, then a set of matching sequences with maximum cost $C$, denoted by $RA_{C}(\omega_{A})$, can be obtained.

The action projection of attacker $\hat{P}: (\Sigma_{o}\cup AT)^* \rightarrow \Sigma_{o}^*$ is defined as:
\[
\hat{P}(\sigma_{R})=\left\{
\begin{aligned}
\sigma_{R}~&~ \text{if} ~\sigma_{R}\in \Sigma_{o}, \\
\sigma_{oi}~&~ \text{if} ~\sigma_{R}=d_{\sigma_{oi}}\in D, \\
\varepsilon~~~&~ \text{if} ~\sigma_{R}=i_{\sigma_{oj}}\in I\cup\{\varepsilon\}, \\
\sigma_{oi}~&~ \text{if} ~\sigma_{R}=t_{\sigma_{oi}\sigma_{oj}}\in T, \\
\end{aligned}
\right.
\]	
and $\hat{P}(\omega_{R}\sigma_{R})=\hat{P}(\omega_{R})\hat{P}(\sigma_{R})$ for $\omega_{R}\in(\Sigma_{o}\cup AT)^*$, $\sigma_{R}\in \Sigma_{o}\cup AT$.
For simplicity, $\hat{P}$ is also used to project $(\Sigma_{o}\cup AT)^*\times \mathbb{N} \rightarrow \Sigma_{o}^*\times \mathbb{N}$, which is defined as for $\omega_{R}\in(\Sigma_{o}\cup AT)^*$, $c\in \mathbb{N}$, $\hat{P}((\omega_{R},c)):=(\hat{P}(\omega_{R}),c)$.

\begin{example}\label{recoveried sequences with cost}
	Consider the {string} $\omega=\zeta\alpha\alpha\alpha$ and $A(\omega)$ already discussed in Example~\ref{corrupted sequences with cost}. Assume that the attacker corrupts $\omega$ to $\omega_{A}=\beta\alpha\alpha \in A(\omega)$. {We have}  $RA(\omega_{A})={d_\alpha}^*(\beta+t_{\alpha\beta}+i_\beta){d_\alpha}^*(\alpha+t_{\gamma\alpha}){d_\alpha}^*(\alpha+t_{\gamma\alpha}){d_\alpha}^*=$
	$\{\beta\alpha\alpha$,
	$\beta t_{\gamma\alpha}\alpha$,
	$\beta\alpha t_{\gamma\alpha}$,
	$i_\beta\alpha\alpha$,
	$t_{\alpha\beta}\alpha\alpha$,
	$\beta t_{\gamma\alpha}t_{\gamma\alpha}$,	
	$d_\alpha\beta\alpha\alpha$,	
	$\beta d_\alpha\alpha\alpha$,	
	$\beta\alpha d_\alpha\alpha$,	
	$\beta\alpha\alpha d_\alpha,...\}$.
	{If we} set the upper bound on the total cost to two, {we can obtain} $RA_2(\omega_{A})=\{(\beta\alpha\alpha,0)$,
	$(\beta t_{\gamma\alpha}\alpha,1)$,
	$(\beta\alpha t_{\gamma\alpha},1)$,
	$(i_\beta\alpha\alpha,2)$,
	$(t_{\alpha\beta}\alpha\alpha,2)$,	
	$(\beta t_{\gamma\alpha}t_{\gamma\alpha},2)\}$. {Note that}
	$\hat{P}(RA(\omega_{A}))=\alpha^\ast(\beta+\alpha+\varepsilon)\alpha^\ast(\alpha+\gamma)\alpha^\ast(\alpha+\gamma)\alpha^\ast$ and
	$\hat{P}(RA_2(\omega_{A}))=\{(\beta\alpha\alpha,0)$,
	$(\beta \gamma\alpha,1)$,
	$(\beta\alpha\gamma,1)$,
	$(\varepsilon\alpha\alpha,2)$,
	$(\alpha\alpha\alpha,2)$,	
	$(\beta\gamma\gamma,2)\}$.
\end{example}

Let the total cost of {an attack that corrupts $\omega$ to $\omega_A$} be $c_{A}\in \mathbb{N}$, and the upper bound on the total cost {to} be $C$. If $c_{A}\leq C$, {we write} $(\omega_{A},c_{A})\in A_{C}(\omega)$. Similarly, let the total cost of attacks incurred at a matching sequence $\omega_{R}\in RA(\omega_{A})$ be $c_{R}\in\mathbb{N}$. Note that
$(\omega_{R},c_{R})\in RA_{C}(\omega_{A})$ if $c_{R}\leq C$. The following corollary is an immediate implication of the above discussions.\\

\begin{corollary}\label{corollary1} 
	Given an observation sequence $\omega\in \Sigma^{\ast}_{o}$, suppose that $\omega_{A} \in A(\omega)$ is generated by an attacker by investing $c_{A}$ units. The sequence estimation unit calculates the set of matching sequences $RA(\omega_{A})$. $RA_C(\omega_{A})$ is the set of matching sequences with maximum cost {$C\geq c_R$}. In $RA_C(\omega_{A})$, 
	$ \omega_{A}\in RA(\omega_{A})$ with zero cost and there exists $\omega_{R} \in RA(\omega_{A})$ with cost $c_{A}$, $c_{A}\leq C$, such that $\omega = \hat{P}(\omega_{R})$.
\end{corollary}

The proof of the following proposition follows directly from the definitions and {thus} it is omitted. {An illustration of the setting described in the proposition can be found in Fig.~\ref{AOmegaRAOmega}.}
\begin{figure}[htbp]
	\begin{minipage}[l]{1\linewidth}
		\centering
		\includegraphics[scale=0.43]{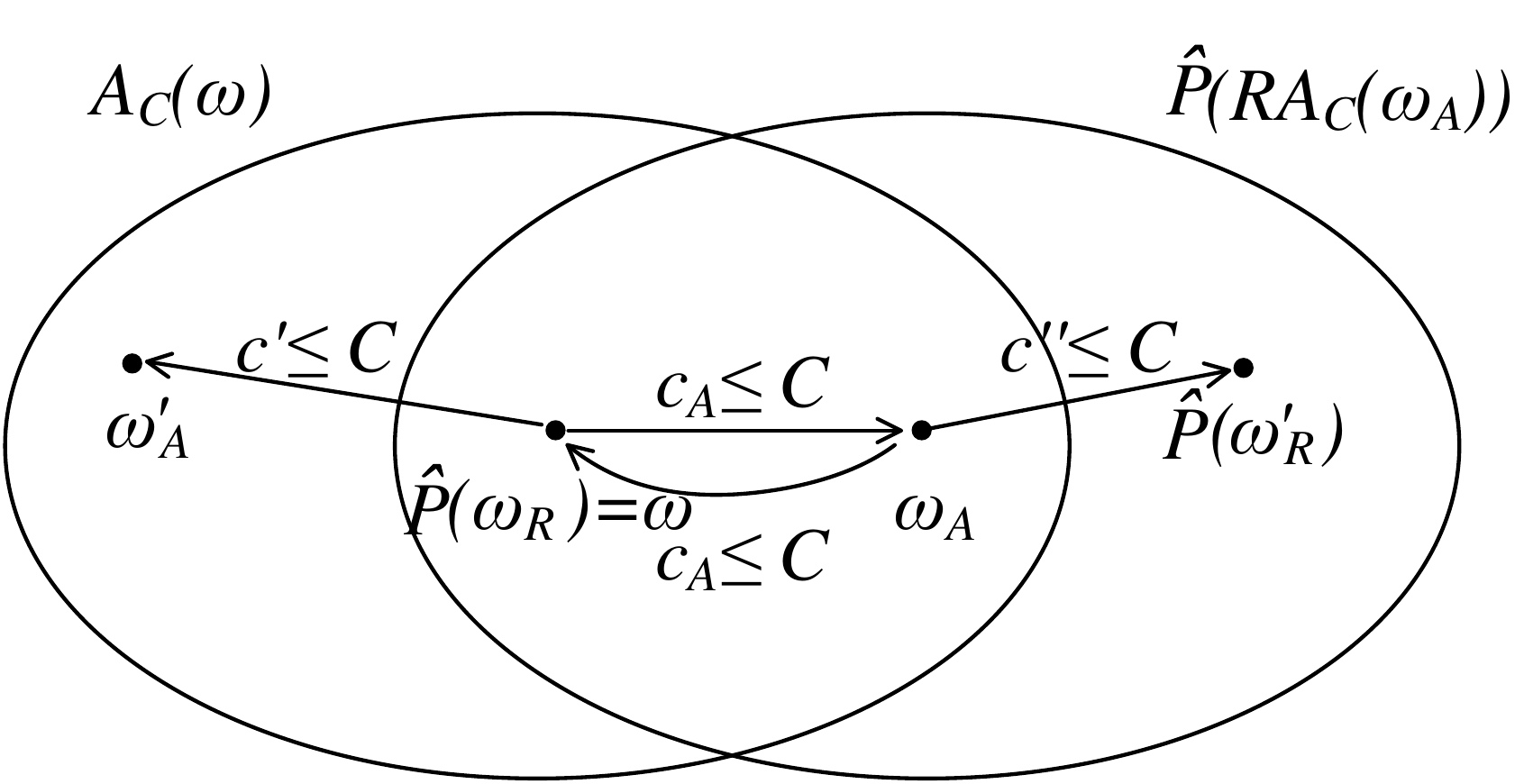}
		
		\caption{Relationship between $A_C(\omega)$ and  $RA_C(\omega_{A})$.}
		\label{AOmegaRAOmega}
	\end{minipage}
	
\end{figure}

\begin{proposition}
	Given a set of tampered sequences $A(\omega)$, a set of matching sequences $RA(\omega_{A})$, {let} $A_C(\omega)$ and $RA_C(\omega_{A})$ respectively denote the sets of {tampered and matching sequences} with upper bound $C$ on the total cost.
	
	1) For all $\omega_{A}\in A(\omega)$, there exists $\omega_{R}\in RA(\omega_{A})$ such that $\omega=\hat{P}(\omega_{R})$.
	
	2) For all $(\omega_{A},c_{A})\in A_{C}(\omega)$, there exists $(\omega_{R},c_{A})\in RA_{C}(\omega_{A})$ such that $(\omega,c_{A})=(\hat{P}(\omega_{R}),c_{A})$.
\end{proposition}

\section{Least Cost State Estimation under Attacks}\label{StateEstimationAttack}
\begin{StateEstimationAttack}
	Consider a DES {modeled} by an NFA $G_{nd}$ and a sensor measuring unit {able to measure and report the sequence of observable events $\omega$ to an observer. An attacker may intercept and alter, at a certain cost, symbols in the reported sequence of observations. Given an upper bound on the total cost $C$ the attacker incurred, the}  observer needs to estimate possible states {(and their associated costs) according to the possibly corrupted sequence $\omega_A$} received from the sequence estimation unit {(refer to Fig.~\ref{InternalStructure})}.
\end{StateEstimationAttack}

We now argue that the set of matching sequences $RA(\omega_{A}
)$ can be  described by the language of an observation automaton, denoted by $G_{s}=(X_{s}, \Sigma_{o}\cup AT, \delta_{s}, x_{0,s})$, where {$X_{s}= \{ 0, 1, 2, ..., | \omega_A | \}$ is the} set of states {(which can be thought as observation stages)}, $\delta_{s}: X_{s}\times (\Sigma_{o}\cup AT) \rightarrow X_{s}$ {is the state transition function, and {$x_{0,s}=0\in X_{s}$}} is the initial state. Suppose that $\omega_{A}=\sigma_{A1}\sigma_{A2}...\sigma_{Am}\in \Sigma_{o}^\ast$ and $RA(\omega_{A})=D^*(\sigma_{A1}+\sigma_{A1}')D^*(\sigma_{A2}+\sigma_{A2}')D^*...D^*(\sigma_{Am}+\sigma_{Am}')D^*\subseteq (\Sigma_{o}\cup AT)^\ast$ (see the definition of the set of all matching sequences). For $x_{i,s}\in X_s$, $i\in\{0,1,2,...,|\omega_{A}|\}=\{0,1,2,...,m\}$, and $\sigma_R \in \Sigma_{o}\cup AT$, the state transition function $\delta_{s}$ is defined as: 
\[\delta_{s}(x_{i,s},\sigma_R)=\left\{
\begin{aligned}
x_{i,s}~~~~&~ \text{if} ~\sigma_{R} \in D, \\
x_{{i+1},s}~ &~ \text{if} ~\sigma_{R}=\sigma_{Ai+1}\vee \sigma_{Ai+1}',
\end{aligned}
\right.\]
where $x_{{i+1},s}\in X_s$ represents the observation stage subsequent to $x_{i,s}$.

If we set the upper bound on the total cost to $C+1$, {we argue that}
$RA_{C+1}(\omega_{A})$ can be described by the language of a DFA, denoted by $G_{sc}(C+1)=(X_{sc}, \Sigma_{o}\cup AT, \delta_{sc}, x_{0,sc})$, where $X_{sc}\subseteq X_s \times \{ 0, 1, 2,..., C+1 \}$ is a set of states with costs,  $x_{0,sc}=(0,0)$ is the initial state, and
$\delta_{sc}:X_{sc} \times (\Sigma_{o}\cup AT) \rightarrow X_{sc}$ is the state transition function, defined as follows: for $(x_{s}, c_{s})\in X_{sc}$  and $\sigma_{R}\in \Sigma_{o}\cup AT$, {we have}
$\delta_{sc}((x_{s}, c_{s}),\sigma_{R})=
(\delta_{s}(x_{s},\sigma_{R}),\min(c_{s}+\Pi_{c}(\sigma_{R}),C+1))$ (undefined if $\delta_s(x_s, \sigma_R)$ is undefined).

The state transition function can be extended to the domain $X_{sc} \times (\Sigma_o \cup AT)^*$ in the standard recursive manner:
$\delta_{sc}((x_{s}, c_{s}),\varepsilon)=(x_{s}, c_{s})$, and $\delta_{sc}((x_{s}, c_{s}),\sigma_{R}\omega_{R})=\delta_{sc}((\delta_{sc}(x_{s}, c_{s}),\sigma_{R}),\omega_{R})$ for {$\sigma_R \in (\Sigma_o \cup AT)$}, $\omega_{R}\in(\Sigma_{o}\cup AT)^*$.

DFA $G_{sc}(C+1)$ has a special structure, which becomes more apparent if we
draw {states of the form} $(x_{s}, 0),(x_{s}, 1),...,$ $(x_{s}, C+1)$, for $x_{s}\in X_{s}$, in a column and {states of the} form $(x_{s1}, c),(x_{s2}, c),..., (x_{s|X_{s}|}, c)$, for $c\in\{0,1,...,C+1\}$ in a row.
{We will also} call each column of $G_{sc}(C+1)$ a stage to reflect the notion of the observation step since each forward transition corresponds to a new observation. We illustrate this via the following example.

\begin{example}\label{ExampleSequence}
	{Continuing Example 4, consider} a possibly tampered sequence $\omega_{A}=\beta\alpha\alpha\in A(\omega)$ and {the} set of all matching sequences  $RA(\omega_{A})={d_\alpha}^*(\beta+t_{\alpha\beta}+i_\beta){d_\alpha}^*(\alpha+t_{\gamma\alpha}){d_\alpha}^*(\alpha+t_{\gamma\alpha}){d_\alpha}^*$.
	We {can} describe the set of all matching sequences  {{using {$G_s$} as {shown}}} in Fig.~\ref{AllSequences}. {If} the upper bound on the total cost {satisfies $C+1=3$, and costs are {given} as in Table~\ref{attacksCost}, the} automaton {$G_{sc}(C+1)$} is {portrayed} in Fig.~\ref{AllSequencescost}. Note that the states with shadow {cannot be} reached in $G_{sc}(C+1)$ {from the initial state} (and can safely be ignored). The initial state of  $G_{sc}(C+1)$ is $(0,0)$ {since initially} the observation automaton is at Step 0 with zero cost. If $\beta$ is observed at the sequence estimation unit, $G_s$ goes to Step 1 with zero cost. If $\beta$ is inserted by the attacker, $G_s$ goes to Step 1 with two units of costs, i.e., $i_\beta$ leads {from} state $(0,0)$ to state $(1,2)$. All reachable states of $G_{sc}(C+1)$ are limited to have maximum three units of costs. {For} instance, $t_{\alpha\beta}$ leads {from} state $(0,3)$ to state $(1,3)$ instead of state $(1,5)$.

	\begin{figure}[htb]
		\begin{minipage}[t]{1\linewidth}
			\centering
			\includegraphics[scale=0.95]{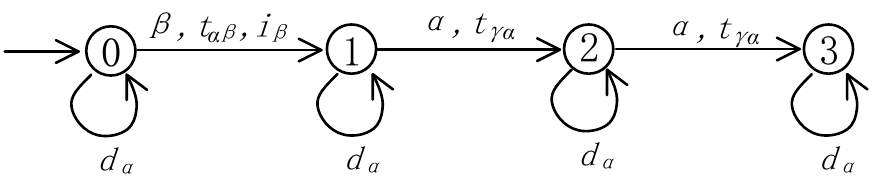}
			
			\caption{Automaton $G_s$ of all matching sequences.}\rm
			\label{AllSequences}
		\end{minipage}
	\end{figure}

	\begin{figure}[htb]
		\begin{minipage}[t]{1\linewidth}
			\centering
			\includegraphics[scale=0.85]{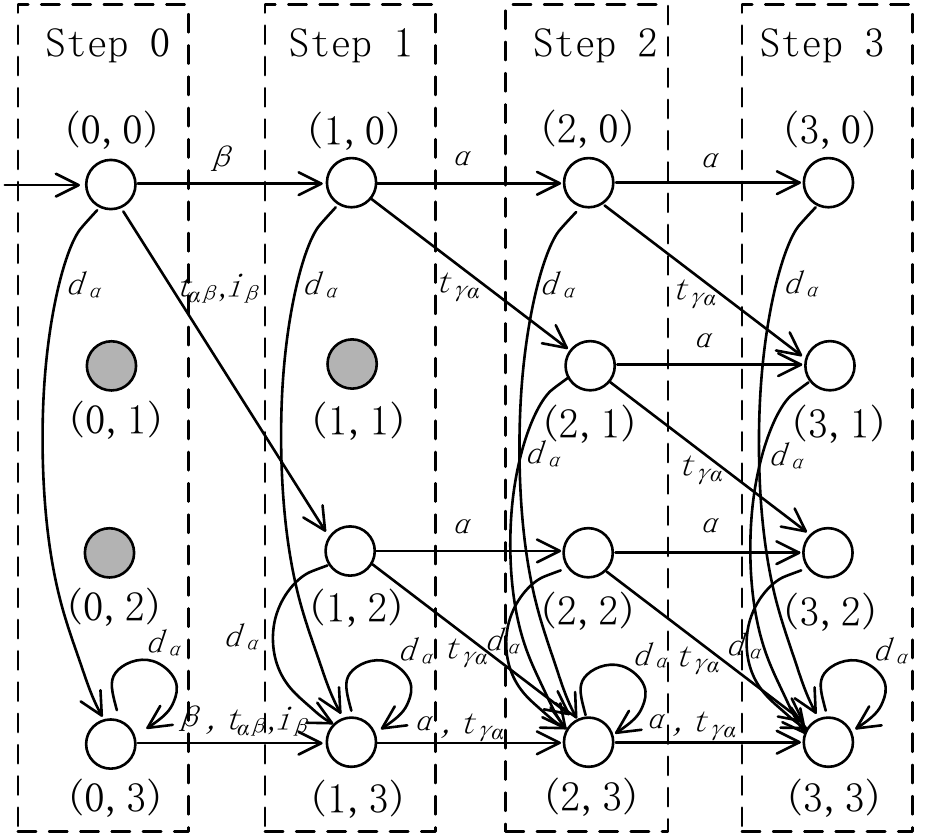}
			
			\caption{Automaton $G_{sc}$ of all matching sequences with costs.}\rm
			\label{AllSequencescost}
		\end{minipage}
	\end{figure}

\end{example}

{Now,} {let us} consider a certain type of parallel operation of $G_{nd}=(X, \Sigma, \delta, X_0)$ and $G_{sc}(C+1)=(X_{sc}, \Sigma_{o}\cup AT, \delta_{sc}, x_{0,sc})$, i.e., by constructing an NFA $H=AC(G_{nd}||G_{sc}(C+1))$ \cite{cassandras2009introduction}, {where $AC(G_{nd}||G_{sc}(C+1))$ represents the {accessible part of {a special type of}} {synchronous composition} of $G_{nd}$ and $G_{sc}(C+1)$.} 
This finite automaton is denoted by $H=(X_{f}, \Sigma_{o}\cup AT, \delta_{f}, X_{0,f})$, where  $X_{f}\subseteq X\times X_{sc}$ is the set of states, $X_{0,f}=\{(x_0,x_{0,sc})|x_0\in X_0\}\subseteq X_{f}$ is a set of initial states, and $\delta_{f}:X_{f} \times (\Sigma_{o}\cup AT) \rightarrow 2^{X_{f}}$ is the state transition function, defined as:
$\delta_{f}((x_{i}, x_{scj}),\sigma_{R})$
\[=\left\{
\begin{aligned}
R(\{x_{i}\},\sigma_{R})\times \{\delta_{sc}(x_{scj},\sigma_{R})\}~ ~~~~&~ \text{if} ~\sigma_{R} \in \Sigma_{o}, \\
R(\{x_{i}\},\hat{P}(\sigma_{R}))\times \{\delta_{sc}(x_{scj},\sigma_{R})\}~ &~ \text{if} ~\sigma_{R} \in AT,
\end{aligned}
\right.
\] 
\noindent where $x_{i} \in X$, $x_{scj} \in X_{sc}$, $\sigma_{R}\in \Sigma_{o}\cup AT$, and $AT= D\cup I\cup T$. 
The domain of
$\delta_{f}$ can be extended to $X_{f} \times (\Sigma_{o}\cup AT)^*$ in the usual way, i.e., for $x_{f}\in X_{f}$, $\omega_{R}\in  (\Sigma_{o}\cup AT)^*$, $\sigma_{R}\in  (\Sigma_{o}\cup AT)$, we have $\delta_{f}(x_{f},\sigma_{R}\omega_{R})=\delta_{f}(\delta_{f}(x_{f},\sigma_{R}),\omega_{R})=\cup_{x_{f}'\in \delta_{f}(x_{f},\sigma_{R})}\delta_{f}(x_{f}',\omega_{R})$.

We construct a \emph{reduced-state} version of $H$, denoted {as} $RH$, by only {maintaining} $X_{Rf}:=\{(x_{i}, x_{s},c)\in X_{f}|\nexists (x_{i}, x_{s},c')\in X_{f}, c'< c\}$ and related transitions. 
$RH$ is defined as a four-tuple NFA $RH=AC(X_{Rf}, \Sigma_{o}\cup AT, \delta_{f}, X_{0,f})$ in the usual way.

The reduced-state version of the parallel composition $H$ can be depicted similarly as $G_{sc}(C+1)$: {states of the form} $(x_{1}, p_{x_{sj}}),(x_{2}, p_{x_{sj}}),..., (x_{|X|}, p_{x_{sj}})$, for $p_{{x}_{sj}}\subseteq X_{sc}$, $x_{sj}\in X_{s}$ {appear} in a column and {states of the} form $(x_{i}, p_{x_{s1}}),(x_{i}, p_{x_{s2}}),..., (x_{i}, p_{x_{{s}|\omega_{A}|}})$, for $x_{i}\in X$ {appear} in a row, where $p_{x_{sj}}:=\{(x_{sj},c)|c\in \mathbb{N}, c\leq C+1\}$. This is clarified in the example below.

\begin{example}
	{Consider again} the system in Fig.~\ref{NFA} as in Examples 1--5.
	The reduced state transition cost diagram of $H$ is shown in Fig.~\ref{RstepObserver}. 
	Since $X_0=\{0,1,2,3,4\}$, $RH$ starts at {Step 0 with} initial column $(0,0,0),(1,0,0),(2,0,0),(3,0,0),(4,0,0)$. If the state estimation unit observes $\beta$, $RH$ reaches Step 1. The original event observed at the sensor measuring unit can be $\beta, \varepsilon, \alpha$. 	At state $(1,0,0)$, $t_{\alpha\beta}$ is also feasible and reaches states $(2,1,2)$ and $(3,1,2)$ in $H$.  Since there exist states $(2,1,0)$ and $(3,1,0)$ {with lower costs}, $\delta_{f}((1,0,0),t_{\alpha\beta})$ {does not appear in} $RH$ (marked with a dotted line). {Similarly, since states (2,1,0) and (3,1,0) appear, the transitions $i_{\beta}$ do not appear from (2,0,0) and (3,0,0).} \\
	\begin{figure*}
		\begin{minipage}[htb]{1\linewidth}
			\centering
			\includegraphics[scale=0.4]{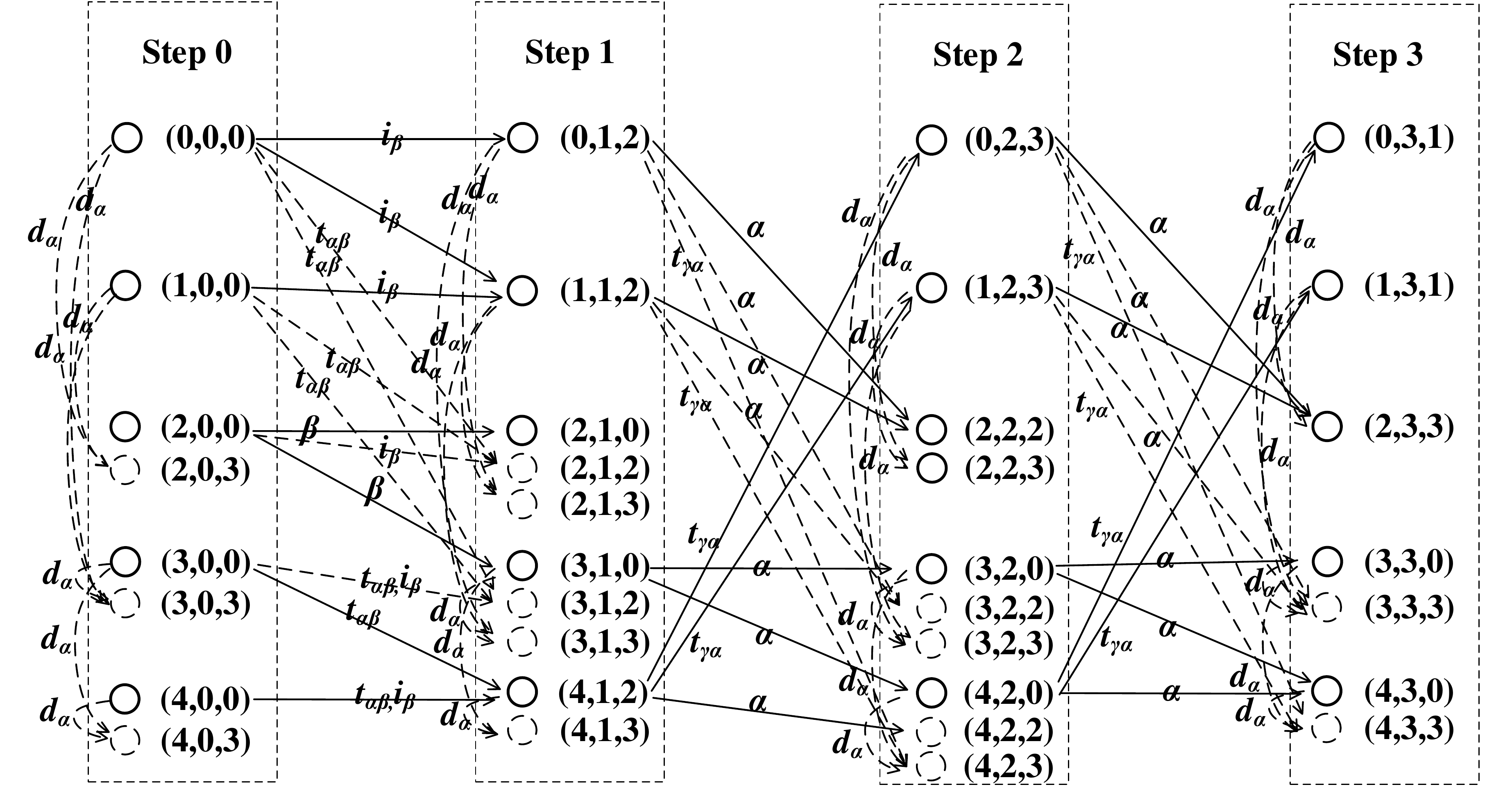}
			
			\caption{The reduced-state version of $H$.}\rm
			\label{RstepObserver}
		\end{minipage}
	\end{figure*}
	
\end{example}

\begin{definition}
	Given a possibly tampered sequence $\omega_{A}\in \Sigma_{o}^\ast$, a set of possible {final (ending)} states in $H$ with the least cost is defined as:
	$E_{H}(\omega_{A})=\{(x_{i},c)|(\exists(x_{i},x_{{s}|\omega_{A}|},c)\in  X_{f}) (\nexists (x_{i},x_{{s}|\omega_{A}|},c')\in  X_{f}) ~c'<c\}$.
\end{definition}

\begin{definition}
	Given a possibly tampered sequence $\omega_{A}\in \Sigma_{o}^\ast$, a set of possible {final (ending) states} in $RH$ with the least cost is defined as:
	$E_{RH}(\omega_{A})=\{(x_{i},c)|\exists(x_{i},x_{{s}|\omega_{A}|},c)\in  X_{Rf}\}$.
\end{definition}

\begin{proposition}\label{Pro2}
	Given a possibly tampered sequence $\omega_{A}=\sigma_{A1}\sigma_{A2}...\sigma_{Am}\in \Sigma_{o}^{\ast}$, where $\sigma_{Ai}\in \Sigma_{o}$ and $i\in\{1,2,...,m\}$, for any estimated state with the least cost $(x_i,c)\in E_{H}(\omega_{A})$, $(x_i,c)\in E_{RH}(\omega_{A})$ {holds}.
\end{proposition}

\begin{proof}
	{By contradiction}, suppose that there is a state $(x'_i,c')\in E_{H}(\omega_{A})$ and $(x'_i,c')\notin E_{RH}(\omega_{A})$. Let $(x_e,x_{{s}j}, c_1)\in X_f$, $\omega_R\in RA(\sigma_{Aj}\sigma_{A(j+1)}...\sigma_{Am})$, $j\in\{1,2,...,m\}$, and $\delta_{f}((x_e,x_{{s}j}, c_1),\omega_R)=(x'_i,x_{{s}|\omega_{A}|},c')$. Suppose that $(x_e,x_{{s}j}, c_1)$ is deleted while calculating $X_{RH}$. {This means that there exists}  $(x_e,x_{{s}j}, c'_1)\in X_f$ such that $c'_1<c_1$. Let the total cost of $\omega_R$ be $c_s$. If $\delta_{f}((x_e,x_{{s}j}, c_1), \omega_R)!$, $\delta_{f}((x_e,x_{{s}j}, c'_1),\omega_R)$ must be defined and can lead to the same state of plant $x'_i$. Since $c'_1+c_s<c_1+c_s=c'$, $(x'_i,x_{{s}|\omega_{A}|},c'_1+c_s)\in E_{H}(\omega_{A})$, which {contradicts the definition of $E_{H}$}.  
\end{proof}
We now formulate an algorithm for finding possible states with respect to the least-cost sequence.
\begin{algorithm}[htbp] 
	\caption{ Least cost state estimation} 
	\label{A2} 
	\begin{algorithmic}[1] 
		\REQUIRE 
		An NFA $G_{nd}=(X, \Sigma, \delta, X_0)$ and a possibly tampered  sequence $\omega_{A}=\sigma_{A1}\sigma_{A2}...\sigma_{Am}$, where $\sigma_{Ai}\in \Sigma_{o}$ for $i=1,2,...,m$.
		\ENSURE
		A set of states $E_{RH}(\omega_{A})$ with the least cost.
		\STATE Calculate the set of matching sequences $RA(\omega_{A})$ and obtain the observation automaton $G_{s}$;
		\STATE Calculate the set of matching sequences $RA_{C+1}(\omega_{A})$ with maximum cost $C+1$ and obtain the finite automaton $G_{sc}(C+1)$;
		\STATE Construct $H=AC(G_{nd}||G_{sc}(C+1))=(X_{f}, \Sigma_{o}\cup AT, \delta_{f}, X_{0,f})$;
		\STATE Generate a reduced-state version $RH$;	
		\RETURN $E_{RH}(\omega_{A})$. 
	\end{algorithmic}
\end{algorithm}

The complexity  of constructing a parallel composition $H$ is $O(|X||\omega_{A}|(C+2))$, where $|X|$ is the {size of} state space of $G_{nd}$ and {$|\omega_{A}|$ equals} the length of the possibly tampered sequence.  {Note that} each state has $(C+2)$ cost values.\\

\begin{theorem}
	Given an NFA $G_{nd}$, a possibly tampered sequence $\omega_{A}\in \Sigma_{o}^{\ast}$, and an upper bound on the  total cost $C+1$, a set of states with least costs can be obtained by Algorithm 1.	
\end{theorem}

\begin{proof}
	The proof is conducted by induction on the length of $\omega_{A}$. More specifically, we establish that, for all prefixes $\omega'_A$ of $\omega_A$ (of length $|\omega'_A|=0, 1, 2,..., |\omega_A|$), we have that if $(x_i,c) \in E_{RH}(\omega'_A)$, then no $(x_i,c')$ with $c'<c$ belongs in $E_{RH}(\omega'_A)$.
	
	1) As a base case, we consider $|\omega_{A}|=0$ (i.e., $\omega_{A}=\varepsilon$) which implies that $RA(\varepsilon)=D^{\ast}$. Consider for some state $(x', x_{0,s},c') \in E_{RH}(\varepsilon)$, any state $(x', x_{0,s},c)$ with $c<c'$. According to Proposition~\ref{Pro2}, we have $(x', x_{0,s},c) \notin X_f$, which means that there does not exist $\omega_{R}\in D^*$ such that $x' \in\delta(x_0, \hat{P}(\omega_{R}))$ and $c=\Pi_c(\omega_{R})$. Hence, state $x'$ with cost $c$ is not reachable, which establishes the base case.
		
	2) Assume that the induction hypothesis holds, i.e., for all sequences $\omega'_A$ of length $|\omega'_{A}|=k$, $k\in \mathbb{N}$, the set of states of least costs is captured by $E_{RH}(\omega'_{A})$.
		
	3) We now prove the same for any sequence of length $\omega_A$ of length $k+1$. Clearly, $\omega_A$ can be written as $\omega'_A \sigma_o$ for some prefix $\omega'_A$ of length $k$ and some observable event $\sigma_o \in \Sigma_o$. 
	
	Consider any state $(x', x_{s(k+1)},c') \in E_{RH}(\omega'_{A}\sigma_o)$ and consider a state $(x', x_{s(k+1)},c)$ with $c<c'$. Let $(x'',c'') \in E_{RH}(\omega'_{A})$ {be the state from which state $(x',x_{s(k+1)},c)$ is reached}. According to Proposition~\ref{Pro2}, we have $(x', x_{s(k+1)},c) \notin X_f$ {implying that} there does not exist $\omega_{R}\in RA(\sigma_o)$ such that $x'\in\delta(x'', \hat{P}(\omega_{R}))$ and $c=\Pi_c(\omega_{R})+c''$. Hence, state $x'$ with cost $c$ is not reachable.	
	This completes the proof of the induction step {and the proof of the proposition}.
\end{proof}

\section{Tamper-Tolerant Diagnosability under Cost Constrained Attacks}\label{VerificationFault}
In this section, we propose {an} {approach} to verify {$C$-constrained} {tamper-tolerant} diagnosability (i.e., the property of the system to allow, under any behavior in the system, diagnosis of all faults after a finite number of observations following the occurrence of the fault). This verification can be achieved with complexity that is polynomial in the size of the system and the total cost.

\begin{definition}\label{modifiedPlant}
	
	Consider an NFA $G_{nd}=(X, \Sigma, \delta, X_{0})$ that generates observations that can be tampered via a set $AT=D\cup I\cup T$ of deletions, insertions, and substitutions, under a maximum cost $C$. The modified system, denoted by $G_{mnd}(C+1)$, is a four-tuple NFA: $G_{mnd}(C+1)=(X_{mn}, \Sigma, \delta_{mn}, X_{0,mn})$, where $X_{mn}\subseteq X \times \{0,1,2,...,C+1\}$ is the set of states, {each associated with its respective cost}. 
	{The set of events} $\Sigma$ is $\Sigma=\Sigma_{o}\cup \Sigma_{uo}$ with $\Sigma_{o}$ being the set of observable events and $\Sigma_{uo}$ being the set of unobservable events {with} $\Sigma_{f}\subseteq \Sigma_{uo}$ {capturing} the set of fault events to be diagnosed.
	{The set of initial states} $X_{0,mn}=\{(x, 0)|x\in X_{0}\}\subseteq X_{mn}$ is {associated with zero initial cost}. {The  state transition function} $\delta_{mn}:X_{mn}\times \Sigma \rightarrow 2^{X_{mn}}$ is { defined as follows:} 
	for $(x,c) \in X_{mn}$, $e \in \Sigma\cup\{\varepsilon\}$, $\sigma_{oi}\in \Sigma_o$, $\delta_{mn}((x,c),e)=N_0\cup N_T\cup N_D\cup N_I$ with\\
	1) the zero cost set $N_0=\delta(x,e)\times \{c\}$  if $e \in \Sigma$,\\
	2) the deletion set \\$N_D=\begin{cases}
	\delta(x,\sigma_{oi})\times\{c_{d_{\sigma_{oi}}}+c\}& \text{if}~(d_{\sigma_{oi}}\in D) \wedge \\&(c_{d_{\sigma_{oi}}}+c\leq C+1),\\
	\\
	\emptyset& \text{otherwise},
	\end{cases}$\\
	3) the insertion set\\
	$N_I=\begin{cases}
	\{x\}\times\{c_{i_{e}}+c\}& \text{if}~(i_{e} \in I) \wedge (c_{i_{e}}+c\leq C+1),\\
	\emptyset& \text{otherwise}.
	\end{cases}$\\
	4) the substitution set\\ 
	$N_T=\begin{cases}
	\delta(x,\sigma_{oi})\times\{c_{t_{\sigma_{oi}e}}+c\}& \text{if}~(t_{\sigma_{oi}e}\in T)\wedge \\& (c_{t_{\sigma_{oi}e}}+c\leq C+1), \\
	\\
	\emptyset& \text{otherwise},
	\end{cases}$\\
	The domain of $\delta_{mn}$ can be extended to $X_{mn}\times\Sigma^*$ in the usual way, i.e., for $x_{mn}\in X_{mn}$, $s\in  \Sigma^*$, $\sigma\in  \Sigma$, we have $\delta_{mn}(x_{mn},\sigma s):=\cup_{x'\in \delta_{mn}(x_{mn},\sigma)}\delta_{mn}(x',s)$.
\end{definition}


\begin{example}
	{In} the NFA in Fig.~\ref{case1}, $\Sigma=\{\alpha,\beta,\gamma,\zeta,\sigma_{f}\}$, $\Sigma_{o}=\{\alpha,\beta,\gamma,\zeta\}$, $\Sigma_{uo}=\Sigma_{f}=\{\sigma_{f}\}$, and $X_0=\{0\}$. {Suppose} that $AT=T$, where $T=\{t_{\alpha\beta},t_{\beta\gamma},t_{\zeta\gamma}\}$ and the cost of attacks is as shown in Table~\ref{AttackCostcase1}. We set $C+1=5$. The  modified NFA in Def.~\ref{modifiedPlant} is shown in Fig.~\ref{modifiedcase1}. For example, at state $(1,0)$, the attacker can spend one unit to change event $\alpha$ to $\beta$, which causes the modified system transition to state $(2,1)$. 
	
	\begin{figure}[htb]
		
		\begin{minipage}{1\linewidth}
			\centering
			\includegraphics[scale=0.35]{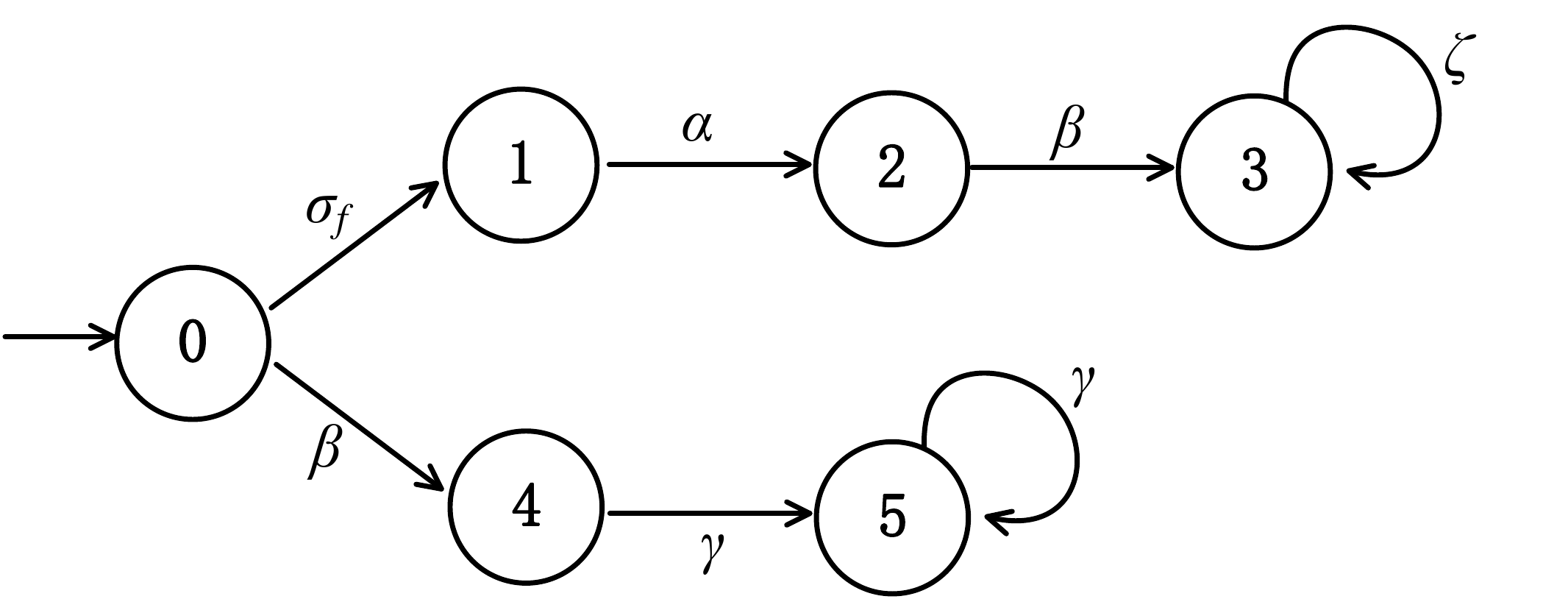}
			
			\caption{Nondeterministic finite automaton with a fault event.}\label{case1}
		\end{minipage}
	\end{figure}
	
	\begin{table}[htb]
		\caption{Attacks with costs for the system in Fig.~\ref{case1}.}\label{AttackCostcase1}
		\begin{tabular} {|c|p{0.61cm}<{\centering}|p{0.61cm}<{\centering}|p{0.61cm}<{\centering}|p{0.61cm}<{\centering}|p{0.61cm}<{\centering}|}
			\hline
			\diagbox{original}{attack}& $\alpha$& $\beta$ & $\gamma$ &$\zeta$& $\varepsilon$ \\
			
			\hline
			$\alpha$& &  1&  &&\\
			\hline
			$\beta$& & & 1&&\\
			\hline
			$\gamma$& & & &&\\
			\hline
			$\zeta$& & & 2 &&\\			
			\hline
			$\varepsilon$& & &&& \\
			\hline				
		\end{tabular}
	\end{table}
	\vspace{10pt}	
	\begin{figure}[htb]
		\begin{minipage}{1\linewidth}
			\centering
			\includegraphics[scale=0.35]{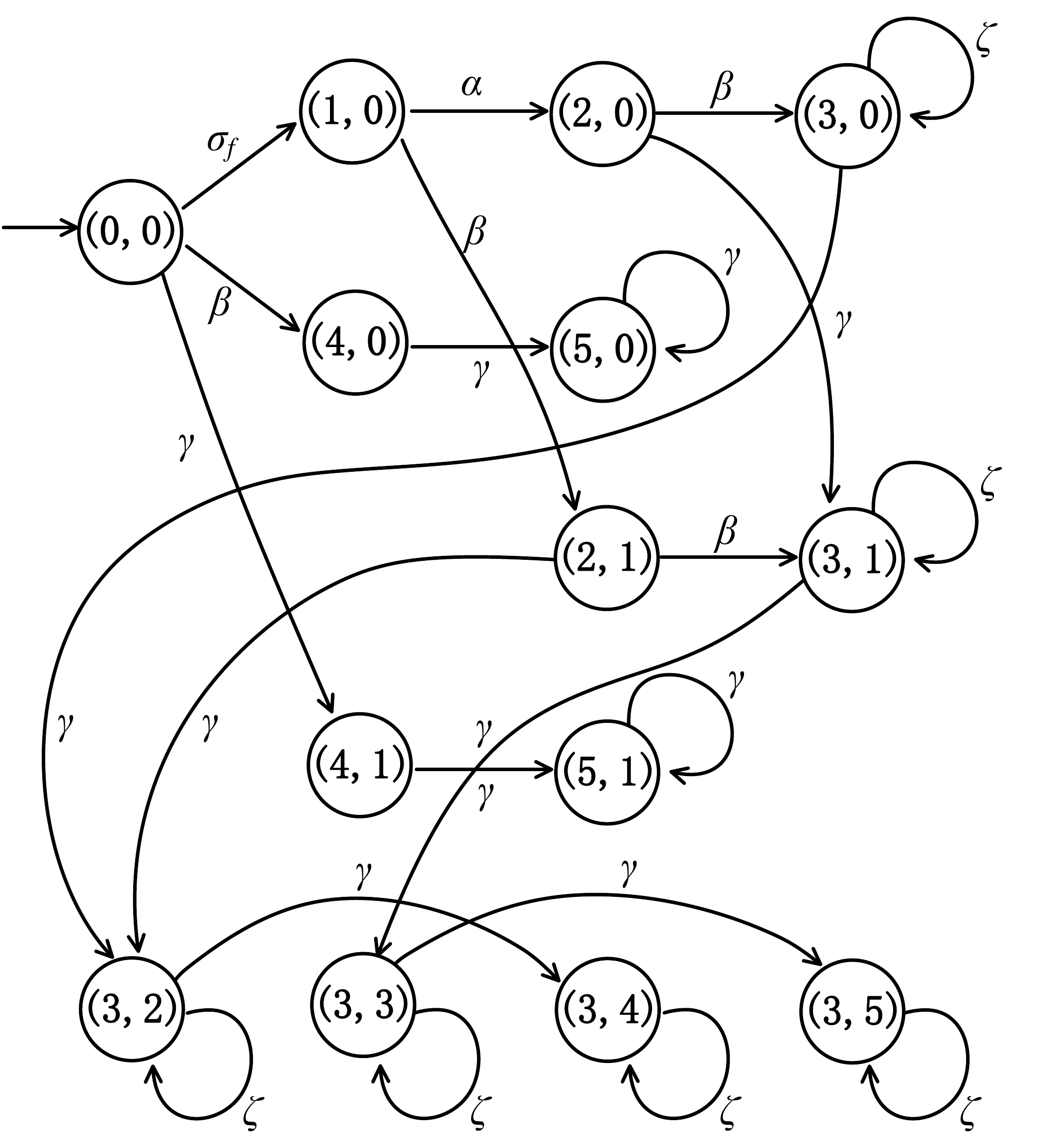}
			
			\caption{Modified NFA for the system in Fig.~\ref{case1}.}\label{modifiedcase1}
		\end{minipage}
		
	\end{figure}	
\end{example}

\begin{example}
	{In} the NFA in Fig.~\ref{case2}, $\Sigma=\{\alpha,\beta,\gamma,\zeta,\sigma_{f}\}$, $\Sigma_{o}=\{\alpha,\beta,\gamma,\zeta\}$, $\Sigma_{uo}=\Sigma_{f}=\{\sigma_{f}\}$, and $X_0=\{0\}$. {Suppose} that $AT=T$, where $T=\{t_{\alpha\gamma},t_{\beta\alpha}\}$. The cost of attacks is shown in Table~\ref{AttackCostFig3}. We set $C+1=3$ and the modified NFA {in} Def.~\ref{modifiedPlant} {is} as shown in Fig.~\ref{modifiedcase2}. For example, the system reaches state $(3,2)$ from state $(1,0)$ if the attacker corrupts $\alpha\beta$ to $\gamma\alpha$.   	
	\begin{figure}[htb]
		
		\begin{minipage}{1\linewidth}
			\centering
			\includegraphics[scale=0.35]{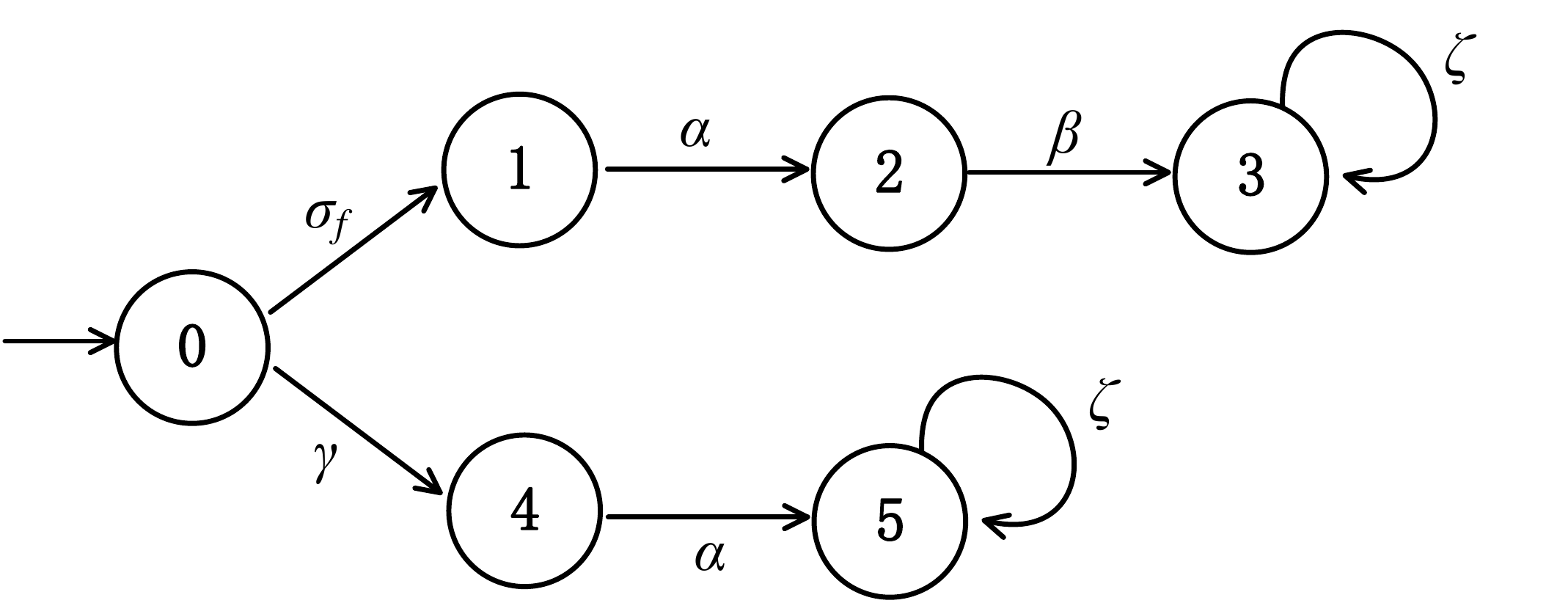}
			
			\caption{Nondeterministic finite automaton with a fault event.}\label{case2}
		\end{minipage}
	\end{figure}

	\begin{table}
		\begin{center}
			\caption{Attacks with costs for the system in Fig.~\ref{case2}.}
			\label{AttackCostFig3}
			\begin{tabular} {|c| p{0.6cm}<{\centering}| p{0.6cm}<{\centering}| p{0.6cm}<{\centering}|p{0.6cm}<{\centering}|p{0.6cm}<{\centering}|}
				\hline
				\diagbox{original}{attack}& $\alpha$& $\beta$ & $\gamma$ &$\zeta$& $\varepsilon$ \\
				
				\hline
				$\alpha$& &  & 1 &&\\
				\hline
				$\beta$&1 & & &&\\
				\hline
				$\gamma$& & & &&\\
				\hline
				$\zeta$& & & &&\\			
				\hline
				$\varepsilon$& & &&& \\
				\hline
				
			\end{tabular}
		\end{center}
	\end{table}
	
	\begin{figure}[htb]
		\begin{minipage}{1\linewidth}
			\centering
			\includegraphics[scale=0.35]{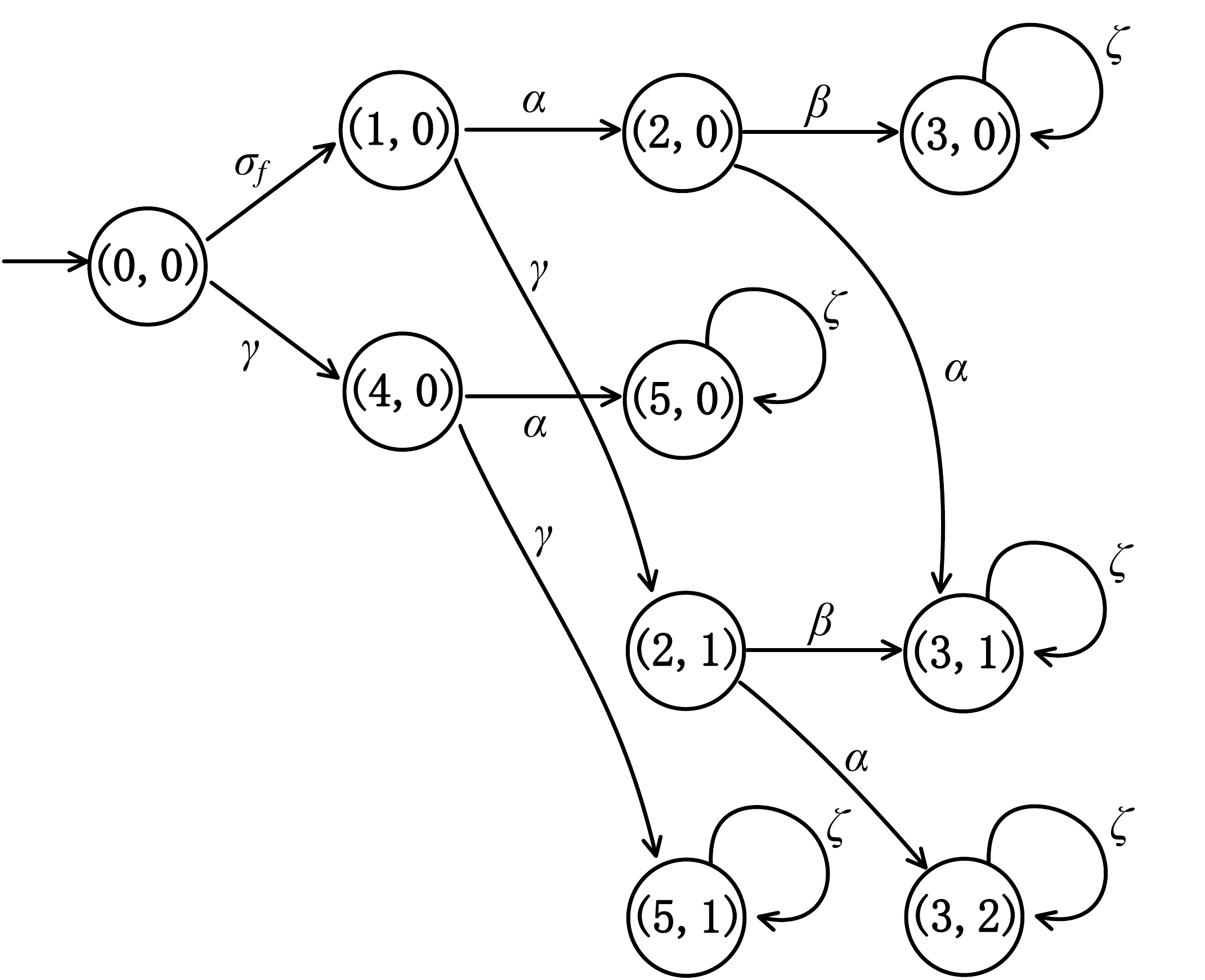}
			
			\caption{Modified NFA for the system in Fig.~\ref{case2}.}\label{modifiedcase2}
		\end{minipage}
		
	\end{figure}

\end{example}

The following assumptions on the language $\mathcal{L}(G_{mnd}(C+1))$ are made when {tamper-tolerant} diagnosability is considered:
(1) We assume as usual the absence in $G_{mnd}$ of cycles of
unobservable events; (2) $\mathcal{L}(G_{mnd}(C+1))$ is live.
\begin{definition}\label{diagnosable}
	An NFA $G_{nd}$ with respect to $\Sigma_{o}$, $\Sigma_{f}$, and $AT$ is said to be $C$\emph{-constrained tamper-tolerant diagnosable} if the following holds: 
	$(\exists n\in \mathbb{N})(\forall s\in (\Sigma\setminus \Sigma_{f})^{\ast}) (\forall\sigma_{f}\in \Sigma_{f}) (\forall s\sigma_{f}\in \mathcal{L}(G_{mnd}(C+1))) (\forall t\in\mathcal{L}(G_{mnd}(C+1))/(s\sigma_{f}))$
	such that $|t|\geq n \Rightarrow \mathcal{D}(s\sigma_{f}t)$, 
	\noindent where the diagnosability function $\mathcal{D}$ is defined as:\\
	\[\mathcal{D}(s\sigma_{f}t)=
	\begin{cases}
	1& \text{if}~[s^{\prime}\in P^{-1}[P(s\sigma_{f}t)]\cap \mathcal{L}(G_{mnd}(C+1))\\&\Rightarrow \sigma_{f}\in s'],\\
	0& \text{otherwise}.\\
	\end{cases}
	\]
\end{definition}

We define the set of possible labels $\Delta =\{N,F\}$, where $N$ {denotes} {normal condition (no failure)}  and $F$ {denotes} that a failure has occurred. {To verify $C$-constrained {tamper-tolerant} diagnosability, we can construct a diagnoser for system $G_{mnd}$ and check for indeterminate cycles as in \cite{sampath1995diagnosability}. Alternatively, we can use a verifier as in \cite{yoo2002polynomial}. Here we follow the latter approach and} construct {{the}} NFA $V_{F}$ for diagnosing the fault events $\Sigma_{f}$ from $G_{mnd}$. We call this automaton the $F$-verifier. {The $F$-verifier is an NFA} $V_{F}=AC(Q^{V_{F}},\Sigma, \delta_{{V}_{F}}, q_{0}^{V_{F}})$,
\noindent where


$Q^{V_{F}}:= X_{mn} \times \Delta \times X_{mn} \times \Delta$

$q_{0}^{V_{F}}:= \cup_{x,y\in X_{0,mn}}\{(x,N,y,N)\}\subseteq Q^{V_{F}}$

For $x_{i}, x_{j}\in X_{mn}$ and $l_{i}, l_{j}\in \Delta$, the (nondeterministic) transition function $\delta_{{V}_{F}}$ is defined as follows.

For $\sigma\in \Sigma_{o}$,
$\delta_{{V}_{F}}((x_{i},l_{i},x_{j},l_{j}),\sigma)=$ $\delta_{mn}(x_{i},\sigma)\times \{l_{i}\} \times\delta_{mn}(x_{j},\sigma)\times \{l_{j}\}.$

For $\sigma\in \Sigma_{uo}\setminus \Sigma_{f}$,
$\delta_{{V}_{F}}((x_{i},l_{i},x_{j},l_{j}),\sigma)=$\\
\[
\begin{cases}
\delta_{mn}(x_{i},\sigma)\times \{l_{i}\}\times \{x_{j}\}\times \{l_{j}\}\\
\{x_{i}\}\times \{l_{i}\} \times\delta_{mn}(x_{j},\sigma)\times \{l_{j}\}\\
\delta_{mn}(x_{i},\sigma)\times \{l_{i}\}\times \delta_{mn}(x_{j},\sigma)\times \{l_{j}\}.
\end{cases}
\]

For $\sigma \in \Sigma_{f}$, $\delta_{{V}_{F}}((x_{i},l_{i},x_{j},l_{j}),\sigma)=$\\
\[
\begin{cases}
\delta_{mn}(x_{i},\sigma)\times \{F\}\times \{x_{j}\} \times \{l_{j}\}\\
\{x_{i}\}\times \{l_{i}\}\times\delta_{mn}(x_{j},\sigma)\times \{F\}\\
\delta_{mn}(x_{i},\sigma)\times \{F\} \times\delta_{mn}(x_{j},\sigma)\times \{F\}.
\end{cases}
\]

Note that for $\sigma\in\Sigma_{o}$, $\delta_{{V}_{F}}((x_{i},l_{i},x_{j},l_{j}),\sigma)$ is {empty} if $\delta_{mn}(x_{i},\sigma)=\emptyset$ or  $\delta_{mn}(x_{j},\sigma)=\emptyset$; for $\sigma\in \Sigma_{uo}\setminus \Sigma_{f}$ or $\sigma \in \Sigma_{f}$, three types of transitions are feasible if $\delta_{mn}(x_{i},\sigma)\neq\emptyset$ and  $\delta_{mn}(x_{j},\sigma)\neq\emptyset$ {whereas only one type of transition is feasible if only one of $\delta_{mn}(x_{i},\sigma)$ or $\delta_{mn}(x_{j},\sigma)$ is non-empty}.
For example, in Fig.~\ref{modifiedcase1}, event $\beta$ is feasible at state $((1,0),F,(0,0),N)$ since $\delta_{mn}((1,0),\beta)$ and $\delta_{mn}((0,0),\beta)$ are both {non-empty}. {In Fig.~\ref{case1}, note that} $\delta_{{V}_{F}}(((1,0),F,(0,0),N),\sigma_f)=\{(1,0)\}\times \{F\}\times\delta_{mn}((0,0),\sigma_f)\times\{N\}$ is also {non-empty} and leads to state $((1,0),F,(1,0),F)$.

A path in {the verifier $V_{F}=AC(Q^{V_{F}},\Sigma, \delta_{{V}_{F}}, q_{0}^{V_{F}})$} is a sequence of states and transitions $\left\langle q_1, \sigma_1, q_2,...,\sigma_{n-1},q_{n}\right\rangle$ such that for each $i\in\{1,2,...,n-1\}$, $q^{V_{F}}_{i+1}\in \delta_{V_{F}}(q^{V_{F}}_{i}, \sigma_{i})$; this path is a cycle if $q_{n}=q_{1}$ and at least one transition is contained along the path.

$V_{F}$ is said to be $F$\emph{-confused} if there is a cycle, $\left\langle q_1, \sigma_1, q_2,...,\sigma_{n-1},q_{n}\right\rangle$, such that for all $q_{i}=(x,l,x',l')$, { $i\in\{1,2,...,n-1\}$, we have} $l=N$ and $l'=F$ or vice versa. If there are no such cycles, we say that $V_{F}$ is $F$\emph{-confusion free}.
\begin{theorem}
	An NFA $G_{nd}$ is $C$-constrained {tamper-tolerant} diagnosable w.r.t. $\Sigma$, $\Sigma_o$, $\Sigma_{f}$, and $AT$ if and only if the corresponding $V_{F}$ is	$F$-confusion free.
\end{theorem}
\begin{proof}
	($\Rightarrow$) Assume that $\mathcal{L}(G_{nd})$ is $C$-constrained {tamper-tolerant} diagnosable w.r.t. $\Sigma$, $\Sigma_o$ and $\Sigma_{f}$. By contradiction, suppose that $V_{F}$ {has an} $F$-confused cycle $\left\langle q_1, \sigma_1, q_2,...,\sigma_{n-1},q_{n}\right\rangle$. Let $q_{1}=(x_{i},N,x_{j},F)$. There {exist} $s,s'\in \mathcal{L}(G_{mnd}(C+1))$, and $x,y\in X_{0,mn}$ such that $P(s)=P(s')$, $x_{i}\in\delta_{mn}(x,s)$, $x_{j}\in\delta_{mn}(y,s')$, $\Sigma_f\in s$, and $\Sigma_f\notin s'$. Now, we have $s(\sigma_1\sigma_2...\sigma_{n-1})^{k}, s'(\sigma_1\sigma_2...\sigma_{n-1})^{k} \in \mathcal{L}(G_{mnd}(C+1))$ with the same projection for $k\geq 0$. 
	{It is} obvious that fault events in $s$ are not diagnosable  {since $k$ can be} arbitrarily large. The definition of $C$-constrained {tamper-tolerant} diagnosability is violated.
	
	($\Leftarrow$) By contrapositive, suppose that $\mathcal{L}(G_{nd})$ is not $C$-constrained {tamper-tolerant} diagnosable w.r.t. $\Sigma$, $\Sigma_o$, $\Sigma_{f}$ and $AT$. {This means that for any nonnegative integer $n$,} we can find $s\in (\Sigma\setminus \Sigma_{f})^{\ast}, \sigma_{f}\in \Sigma_{f}$, such that $s\sigma_{f}\in \mathcal{L}(G_{mnd}(C+1)$ and the following is true: $(\exists t\in\mathcal{L}(G_{mnd}(C+1))/(s\sigma_{f}))$   $\{(|t|\geq n)$ and $(\exists s^{\prime}\in P^{-1}[P(s\sigma_{f}t)]\cap \mathcal{L}(G_{mnd}(C+1)))$ such that $\sigma_{f}\notin s'\}$. Let $l\in \overline{s'}$ and $P(l)=P(s\sigma_{f})$. It is obvious that $\Sigma_{f} \notin l$. Let $x_{s\sigma_{f}}\in \delta_{mn}(x,s\sigma_{f})$, $x_l\in \delta_{mn}(y,l)$, $x_{s\sigma_{f}t}\in \delta_{mn}(x_{s\sigma_{f}},t)$, and $x_{s'}\in \delta_{mn}(x_l,\{s'\}/l)$. We {obtain} reachable states $(x_{s\sigma_{f}},F,x_l,N), (x_{s\sigma_{f}t},F,x_{s'},N) \in Q^{V_{F}}$ in $V_{F}$. {Since $n$ can be arbitrarily large, choose} $n'\geq{(2|X|(C+2))}^{2}$. There exists a path, denoted by $\left\langle q_{k1}, \sigma_{k1}, q_{k2},...,\sigma_{k(n'-1)},q_{kn'}\right\rangle$, where $q_{k1}=(x_{s\sigma_{f}},F,x_l,N)$ and $q_{kn'}=(x_{s\sigma_{f}t},F,x_{s'},N)$. Then, it is certain that there exist {$i, j$ satisfying} $1\leq i<j\leq n'$ such that $(q_{ki},F,q_{ki'},N)=(q_{kj},F,q_{kj'},N)$ since $n'\geq {(2|X|(C+2))}^{2}$ {is greater than the maximum possible number of distinct states in the verifier construction}. {Therefore we have identified an} $F$-confused cycle.
\end{proof}
\begin{example}
	We construct {part of} the $F$-verifiers of the modified NFAs  in Figs.~\ref{modifiedcase1} and \ref{modifiedcase2}, as shown in Figs.~\ref{Verifiercase1} and \ref{Verifiercase2}, respectively.  
	The verifier is $F$-confusion free in Fig.~\ref{Verifiercase1}. Hence, $\Sigma_{f}$ is $C$-constrained {tamper-tolerant} diagnosable for the NFA in Fig.~\ref{case1}. {Note that there can be confusion between} $\beta\gamma\gamma$ and $\sigma_{f}\alpha\beta\zeta$ when the attacker corrupts $\alpha\beta\zeta$ to $\beta\gamma\gamma$. However, diagnosis {is possible since eventually} $\zeta$ {will be} observed without corruption due to {the} limitation of the total cost of attacks. {Since} the verifier in Fig.~\ref{Verifiercase2} is $F$-confused, $\Sigma_{f}$ is not $C$-constrained {tamper-tolerant} diagnosable for the NFA in Fig.~\ref{case2}. For the system in Fig.~\ref{case2}, if the attacker successfully corrupts $\alpha\beta$ to $\gamma\alpha$, $\sigma_{f}$ {cannot be diagnosed} {regardless of how long we wait for additional observations.}

	\begin{figure}[htb]
		\begin{minipage}[l]{1\linewidth}
			\centering
			\includegraphics[scale=0.28]{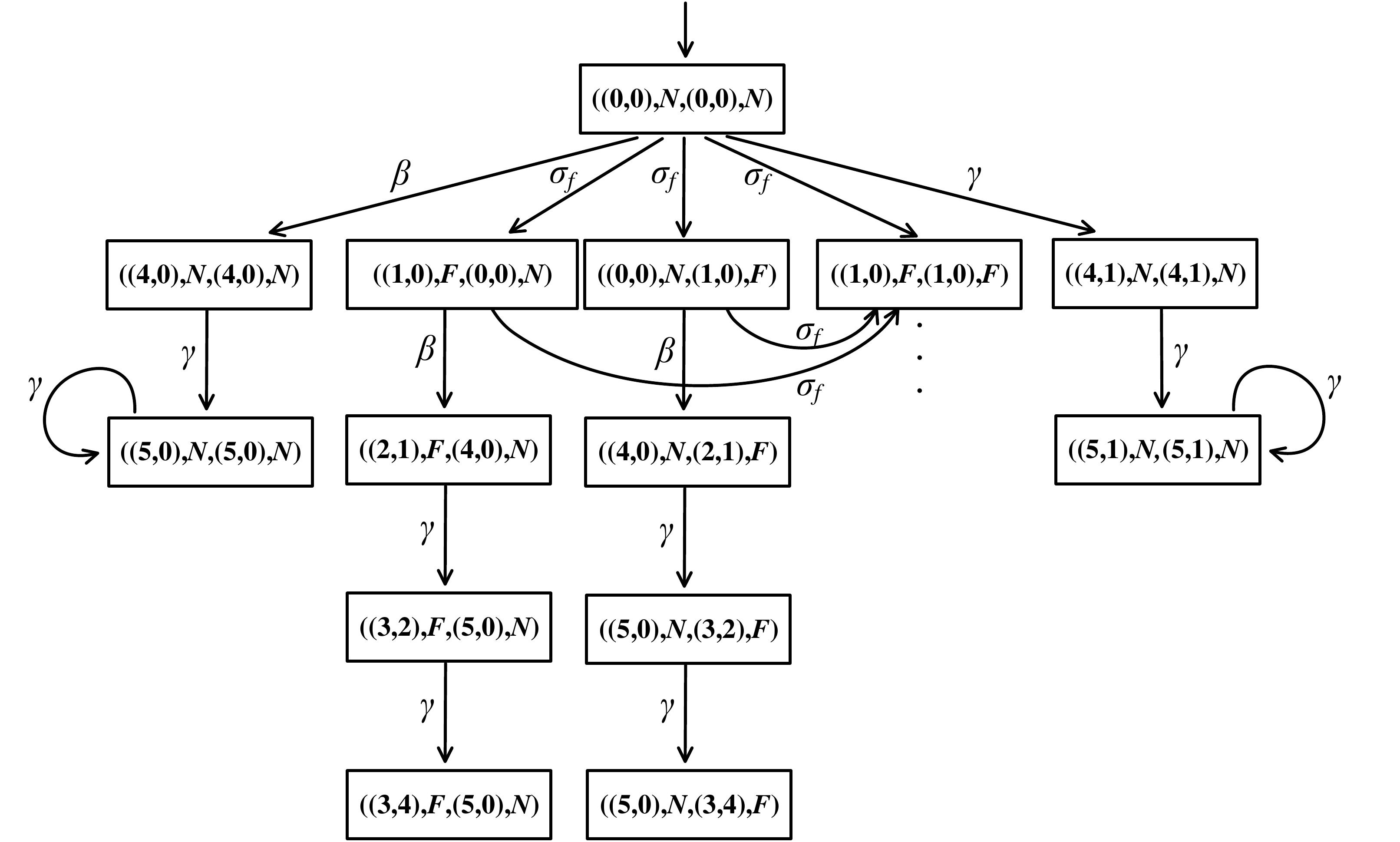}
			
			\caption{Part of the verifier for the modified NFA in Fig.~\ref{modifiedcase1} (continuations not shown cannot lead to $F$-confused cycles).}
			\label{Verifiercase1}
		\end{minipage}
	\end{figure}
	
	\begin{figure}[htb]
		\begin{minipage}[l]{1\linewidth}
			\centering
			\includegraphics[scale=0.4]{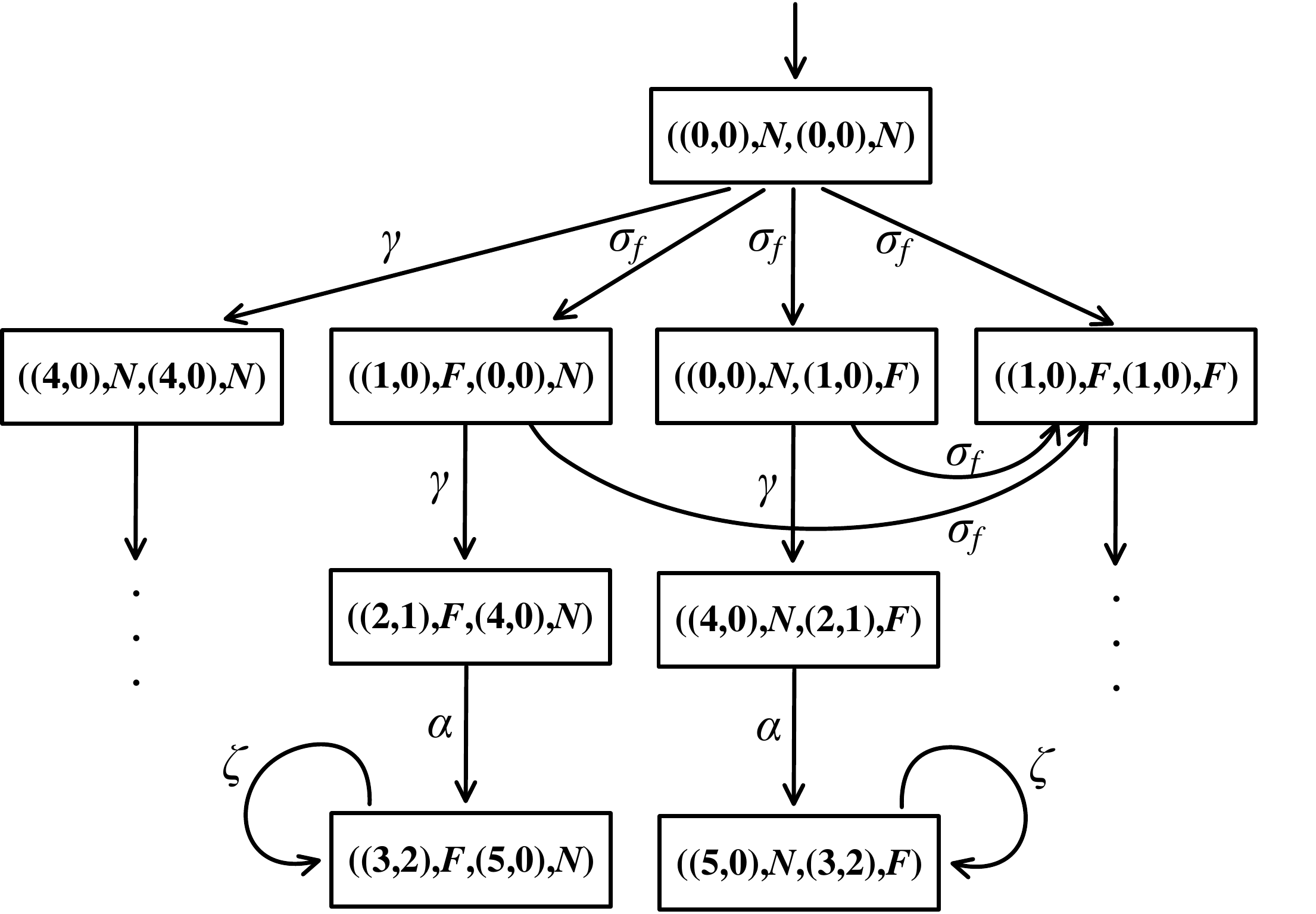}
			
			\caption{Part of the verifier for the modified NFA in Fig.~\ref{modifiedcase2} (continuations not shown cannot lead to $F$-confused cycles).}
			\label{Verifiercase2}
		\end{minipage}
	\end{figure}	
\end{example}

Let {$|X|$ denote} the number of states of $G_{nd}=(X, \Sigma, \delta, X_0)$.
The number of reachable states of $V_{F}$ is {at most} $(2|X|(C+2))^2$. Therefore, the overall complexity {for verifying $C$-constrained {tamper-tolerant} diagnosability using an $F$-verifier} is $O(|X|^2C^2)$.

\section{$C$-Constrained Tampering}\label{Extension}

In this section, we study the case where an attacker, under a constraint of a total cost $C$ on its tampering action, has the capability to cause  a violation of $C'$-constrained tamper-tolerant diagnosability of $G_{nd}$ for arbitrarily large $C'$ ($C'\geq C$). In other words, the attacker can, {at least under some activity in the system,} {coordinate its tampering action to}  keep {the observer indefinitely confused} while utilizing a finite {number} of attacks (more generally, a finite total cost $C$).
Furthermore, we show how one can efficiently calculate the minimum value of $C$ that causes such a violation for at least one fault {within} the behavior of the system.

A useful (and obvious) corollary is presented to explicitly {state} a special case of the existence of $C$.

\begin{corollary}\label{Extencoro}
	If $G_{nd}$ is not diagnosable \cite{sampath1995diagnosability,cassandras2009introduction}, then it is not $C'$-constrained tamper-tolerant diagnosable for any $C'\geq C\geq 0$. 
\end{corollary}

In the case that system $G_{nd}$  {{\em is}} diagnosable, we need to confirm whether the attacker can corrupt the output of the system such that a particular fault does not get diagnosed and remains non-diagnosable indefinitely with a finite number of attacks.

\begin{definition}\label{modifiedPlant2}
	
	Given an NFA $G_{nd}=(X, \Sigma, \delta, X_{0})$, the corrupted system, denoted by $G_{cn}$, is an NFA $G_{cn}=(X, (\Sigma\cup\{\varepsilon\})\times \mathbb{N}, \delta_{cn}, X_{0})$, where $(\Sigma\cup\{\varepsilon\})\times \mathbb{N}$ is the set of pairs involving an event and its corresponding cost.
	{The  state transition function} $\delta_{cn}:X\times ((\Sigma\cup\{\varepsilon\})\times \mathbb{N}) \rightarrow 2^{X}$ is defined as follows: 
	for $x \in X$, $(e,c) \in (\Sigma\cup\{\varepsilon\})\times \mathbb{N}$,
	$\sigma_{oi}\in \Sigma_o$,
	
	$\delta_{cn}(x,(e,c))=\begin{cases}
	\delta(x,e)& \text{if}~c=0,\\	
	\\
	\delta(x,\sigma_{oi})& \text{if}~(e=\varepsilon)\wedge(\sigma_{oi}\in \Sigma_D)\wedge\\& (c=\Pi_{c}(d_{\sigma_{oi}})>0),\\
	\\		
	x& \text{if}~(e\in \Sigma_I)\wedge (c=\Pi_{c}(i_e)>0),\\
	\\
	\delta(x,\sigma_{oi})& \text{if}~((\sigma_{oi},e)\in \Sigma_T) \wedge\\& (c=\Pi_{c}(t_{\sigma_{oi},e})>0).\\		
	\end{cases}$\\

	The domain of $\delta_{cn}$ can be extended to $X\times(\Sigma\cup\{\varepsilon\}\times \mathbb{N})^\ast$ in the usual way, i.e., for $x\in X$, $s_c\in ((\Sigma\cup\{\varepsilon\})\times \mathbb{N})^*$, $(e,c)\in (\Sigma\cup\{\varepsilon\})\times \mathbb{N}$, we have $\delta_{cn}(x,(e,c)s_c):=\cup_{x'\in \delta_{cn}(x,(e,c))}\delta_{cn}(x',s_c)$.
\end{definition}



\begin{example}
	{Considering again} the NFA $G_{nd}$ in Fig.~\ref{case2}, all  transitions defined in $G_{nd}$ are set with zero cost in $G_{cn}$ as shown in Fig.~\ref{FigCorNFA2}. The pairs involving events and positive costs are also partially defined in some of the states according to Table~\ref{AttackCostFig3}, such as $(\gamma,1)$ at state 1.	
	
	\begin{figure}[htbp]
		\begin{minipage}[l]{1\linewidth}
			\centering
			\includegraphics[scale=0.33]{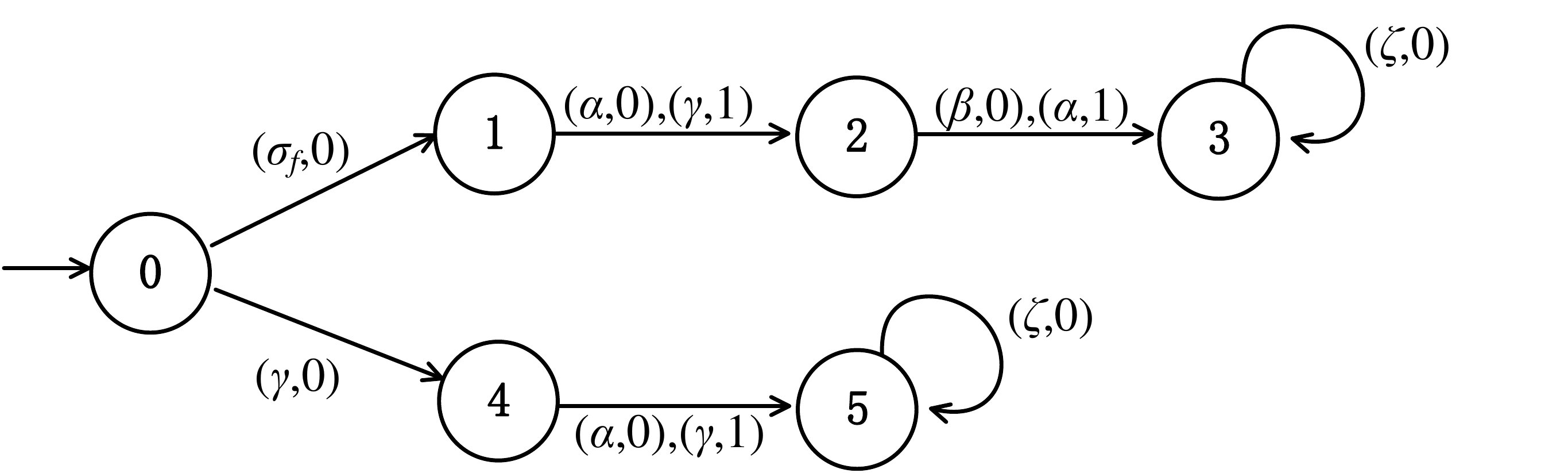}
			\caption{Corrupted automaton $G_{cn}$ for NFA in Fig.~\ref{case2}.}\rm
			\label{FigCorNFA2}
		\end{minipage}
	\end{figure}
\end{example}

The verifier for the corrupted system, denoted by $V'_F$, is called a modified verifier. The automaton $V'_F$ is defined similarly {to} $V_F$, i.e.,  $V'_{F}=AC(Q^{V'_{F}},((\Sigma\cup\{\varepsilon\})\times \mathbb{N})\times((\Sigma\cup\{\varepsilon\})\times \mathbb{N}), \delta_{V'_{F}}, q_{0}^{V'_{F}})$,
\noindent where

$Q^{V'_{F}}:= X \times \Delta\times X \times \Delta$

$q_{0}^{V'_{F}}:= \cup_{x,y\in X_{0}}\{(x,N,y,N)\}\subseteq Q^{V'_{F}}$

For $x_{i}, x_{j}\in X$, $l_{i}, l_{j}\in \Delta$, and $c,c' \in \mathbb{N}$, the (nondeterministic) transition function $\delta_{V'_{F}}$ is defined as follows.

For $\tau=((e,c),(e,c'))\in ((\Sigma_{o}\cup\{\varepsilon\})\times \mathbb{N})\times ((\Sigma_{o}\cup\{\varepsilon\})\times \mathbb{N})$, we define 
$\delta_{V'_{F}}((x_{i},l_{i},x_{j},l_{j}),\tau)=$
\[\delta_{cn}(x_{i},(e,c))\times \{l_{i}\} \times\delta_{cn}(x_{j},(e,c'))\times \{l_{j}\}.\]
For $\tau=((e,0),(e,0))\in ((\Sigma_{uo}\setminus \Sigma_{f})\times \mathbb{N})\times ((\Sigma_{uo}\setminus \Sigma_{f})\times \mathbb{N})$,
$\delta_{V'_{F}}((x_{i},l_{i},x_{j},l_{j}),\tau)=$
\[\begin{cases}
\delta_{cn}(x_{i},(e,0))\times \{l_{i}\}\times \{x_{j}\}\times \{l_{j}\}\\
\{x_{i}\}\times \{l_{i}\} \times\delta_{cn}(x_{j},(e,0))\times \{l_{j}\}\\
\delta_{cn}(x_{i},(e,0))\times \{l_{i}\}\times \delta_{cn}(x_{j},(e,0))\times \{l_{j}\}.
\end{cases}\]
For $\tau=((e,0),(e,0))\in (\Sigma_{f}\times \mathbb{N})\times (\Sigma_{f}\times \mathbb{N})$,
$\delta_{V'_{F}}((x_{i},l_{i},x_{j},l_{j}),\tau)=$
\[\begin{cases}
\delta_{cn}(x_{i},(e,0))\times \{F\}\times \{x_{j}\} \times \{l_{j}\}\\
\{x_{i}\}\times \{l_{i}\}\times\delta_{cn}(x_{j},(e,0))\times \{F\}\\
\delta_{cn}(x_{i},(e,0))\times \{F\} \times\delta_{cn}(x_{j},(e,0))\times \{F\}.
\end{cases}\]

A path in $V'_{F}$ is a sequence of {states and} transitions $\eta$: $\left\langle q^{V'_{F}}_1, \tau_1, q^{V'_{F}}_2,...,\tau_{n-1},q^{V'_{F}}_{n}\right\rangle$ such that for each $i\in\{1,2,...,n-1\}$, $q^{V'_{F}}_{i+1}\in \delta_{V'_{F}}(q^{V'_{F}}_{i}, \tau_{i})$; this path is a cycle if $q^{V'_{F}}_{n}=q^{V'_{F}}_{1}$ and at least one transition is contained along the path. 

{The} modified verifier $V'_{F}$ is said to be \emph{modified} $F$\emph{-confused} if there is a cycle $\xi$: $\left\langle q^{V'_{F}}_1, \tau_1, q^{V'_{F}}_2,...,\tau_{n-1},q^{V'_{F}}_{n}\right\rangle$ such that for all $q^{V'_{F}}_{i}=(x,l,x',l')\in Q^{V'_{F}}$, $\tau_{i}=((e,c),(e,c'))$, $i\in\{1,2,...,n-1\}$, we have  $l=N$ and $l'=F$ or vice versa, and $c=c'=0$. We call this cycle a modified $F$-confused cycle. If there are no such cycles, we say that $V'_{F}$ is \emph{modified} $F$\emph{-confusion free}.
We use $q^{V'_{F}}_{i}\in \xi$ to represent $q^{V'_{F}}_{i}$ belonging to $\xi$, where $\xi\in V'_{F}$ represents a cycle $\xi$ in $V'_{F}$. 

Suppose that there exist $m$ modified $F$-confused cycles, denoted by $\xi_1,\xi_2,...,\xi_m$.
We refer to ending states as the set of states in $V'_{F}$ {that are members of at least one of these modified F-confused cycles}; this set is defined as $X_e=\cup_{z\in\{1,2,...,m\}} X_{ez}$, where  $X_{ez}=\{(x_{i},l_{i},x_{j},l_{j})\in Q^{V'_{F}}|\exists\xi_z\in V'_{F},~\text{such that}~ (x_{i},l_{i},x_{j},l_{j})\in\xi_z\}$.
For $V'_F$ in Fig.~\ref{FigVerNFA2}, the set of ending states $X_e=\{(3,F,5,N),(5,N,3,F)\}$.

A path $\eta$: $\left\langle q^{V'_{F}}_1, \tau_1, q^{V'_{F}}_2,...,\tau_{n-1},q^{V'_{F}}_{n}\right\rangle$ has two total costs, called left and right total costs, respectively. The left total cost, denoted by $CP_l(\eta)$, is defined as $CP_l(\eta)=\Sigma_{i\in \{1,2,...,n-1\}}CP_l(\tau_{i})$, where $\tau_{i}=((e,c),(e,c'))$ and $CP_l(\tau_{i})=c$.
The right total cost, denoted by $CP_r(\eta)$, is defined as $CP_r(\eta)=\Sigma_{i\in \{1,2,...,n-1\}}CP_r(\tau_{i})$, where $\tau_{i}=((e,c),(e,c'))$ and $CP_r(\tau_{i})=c'$.
The total cost of $\eta$, denoted by $CP(\eta)$, is defined as $CP(\eta)=\max(CP_l(\eta),CP_r(\eta))$.
Note that we select the maximum value of left and right total costs since if the upper bound on the total cost from $q^{V'_{F}}_1$ is set to the maximum one, {then the attacker has enough costs to generate the sequence of observations that corresponds to this path, starting from either of two different sequences of actual observations that match the (left and right) costs in the path.}

\begin{corollary}\label{extcoroexistcycle}
	There exists a finite positive integer $C$ such that $G_{nd}$ is not $C'$-constrained tamper-tolerant diagnosable  for arbitrarily large $C'$ ($C'\geq C$) if and only if {the} modified verifier $V'_F$ processes a cycle that is modified $F$-confused.
\end{corollary}

\begin{proof}
	($\Rightarrow$)  Suppose that $G_{nd}$ is not $C$-constrained tamper-tolerant diagnosable, where $C\geq 0$. It is certain that {the} $F$-verifier of $G_{mnd}(C+1)$ contains $F$-confused cycles, where the maximum cost associated with each state is $C+1$, which means that, at most, $C+1$ units of costs are required by {the} attacker to corrupt the output of the system such that a particular fault does not get diagnosed and remains non-diagnosable indefinitely.
	{It follows} that the modified verifier $V'_F$ can process a cycle that is modified $F$-confused and {can be reached by having the attacker invest at most} $C+1$ units of costs.
	
	($\Leftarrow$) If {the} modified verifier $V'_F$ processes a modified $F$-confused cycle, then there must exist a path with {a finite number of} transitions from {a pair of} initial states to a particular ending state in {this modified $F$-confused} cycle in $V'_F$. The total cost of the path is finite due to the finite number of transitions. {Suppose that} the total cost of {this} path {is} $C$. We conclude that there exists a sequence of events $s$ followed by a sequence of events $t^n$, such that the attacker can (i) spend the total cost of at most $C$ to generate a sequence of observations that could be matched to the two sequences that correspond to the path that leads to the modified $F$-confused cycle; (ii) spend a total cost zero to cycle through the modified $F$-confused cycle $n$ times (once for each execution of $t$). Therefore, an attacker can make $G_{nd}$ non $C'$-constrained tamper-tolerant diagnosable, {for any $C' \geq C$}, by generating $F$-confused cycles in the $F$-verifier of $G_{mnd}(C'+1)$.
\end{proof}

\begin{proposition}
	Let the maximum individual cost of each attack be $c_{max}$. The modified verifier $V'_F$ can be constructed in $O(c^2_{max}|X|^2|\Sigma_o|)$.
\end{proposition}
\begin{proof}
	The number of reachable states in $V'_F$ is {at most} $4|X|^2$. For each reachable state, there are $(c_{max}+1)^2(|\Sigma_o|+1)+3|\Sigma_{uo}|$ feasible transitions: (i) For $\sigma \in \Sigma_o \cup \{\varepsilon\}$, there are at most $(c_{max}+1)^2$ kinds of pairs involving event $\sigma$ and positive costs; (ii) For $\sigma \in \Sigma_{uo}$, three $((\sigma,0),(\sigma,0))$ can be defined at a reachable state. The construction of $V'_F$ takes $4|X|^2((c_{max}+1)^2(|\Sigma_o|+1)+3|\Sigma_{uo}|)$ {operations with} overall complexity {of} $O(c^2_{max}|X|^2|\Sigma_o|)$.
\end{proof}

For simplicity, we omit the algorithm of identifying all modified $F$-confused cycles. They can be calculated with polynomial complexity using a depth-first search (DFS). More specifically, we mark each state in $Q^{V'_{F}}$ that is visited; if a state is visited for the second time, then one has a cycle (which can be obtained by tracing back the DFS tree).

\begin{example}
	In Fig.~\ref{FigVerNFA2}, $((\sigma_f,0),(\sigma_f,0))$ leads the modified verifier to states $(1,F,0,N)$, $(0,N,1,F)$, and $(1,F,1,F)$ from the initial sate $(0,N,0,N)$. At $(1,F,0,N)$, $\delta_{cn}(1,(\gamma,1))=\{2\}$ and $\delta_{cn}(0,(\gamma,0))=\{4\}$. Hence, $\delta_{V'_{F}}((1,F,0,N),((\gamma,1),(\gamma,0)))=\{(2,F,4,N)\}$. 
	Since $V'_{F}$ is modified $F$-confused, there exists $C$ such that $G_{nd}$ is not $C'$-constrained tamper-tolerant diagnosable. For simplicity, in the diagrams, we omit the self loops with $((\varepsilon,0),(\varepsilon,0))$ at each state.
	\begin{figure}[htbp]
		\begin{minipage}[l]{1\linewidth}
			\centering
			\includegraphics[scale=0.36]{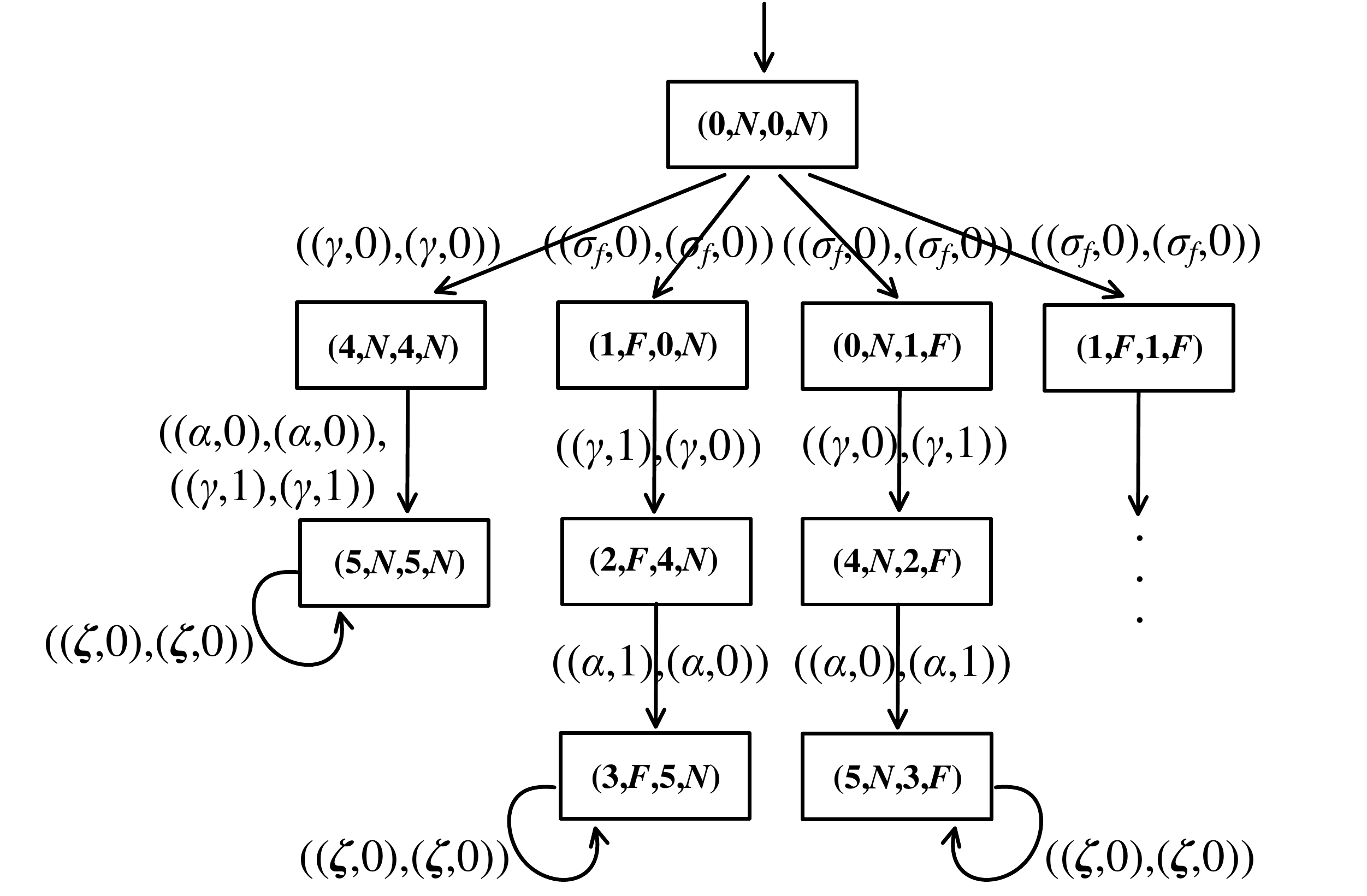}
			\caption{Modified verifier for corrupted system in Fig.~\ref{FigCorNFA2} (continuations not shown cannot lead to modified $F$-confused cycles).}\rm
			\label{FigVerNFA2}
		\end{minipage}
	\end{figure}
\end{example}

\begin{definition}
	The set of paths from an initial state to an ending state in a  modified $F$-confused cycle $\xi_z$, denoted by $Y(\xi_z)$, is defined as $Y(\xi_z)=\{\eta:\left\langle q^{V'_{F}}_1, \tau_1, q^{V'_{F}}_2,...,\tau_{n-1},q^{V'_{F}}_{n}\right\rangle|q^{V'_{F}}_1\in q_{0}^{V'_{F}},q^{V'_{F}}_{n}\in X_{ez}\}$.
\end{definition}

\begin{definition}	
	A path $\eta\in Y(\xi_z)$ is said to be  minimum cost with respect to a modified $F$-confused cycle $\xi_z$ if there does not exist $\eta'\in Y(\xi_z)$ such that $CP(\eta')<CP(\eta)$.	
\end{definition}

For $V'_F$ in Fig.~\ref{FigVerNFA2}, $\langle(0,N,0,N),$ $((\sigma_f,0),(\sigma_f,0)),$ $(1,F,0,N),$ $((\gamma,1),$ $(\gamma,0)),$ $(2,F,4,N),$ $((\alpha,1),$ $(\alpha,0)),$ $(3,F,5,N)\rangle$ and $\langle(0,N,0,N),$ $((\sigma_f,0),$ $(\sigma_f,0)),$ $(0,N,1,F),$ $((\gamma,0),$ $(\gamma,1)),$ $(4,N,2,F),$ $((\alpha,0),$ $(\alpha,1)),$ $(5,N,3,F)\rangle$ are two minimum cost paths.

The minimum value of $C$, denoted by $C_{min}$, can be calculated by $C_{min}=\{CP(\eta)|(\forall z_1,z_2\in\{1,2,...,m\})$ $(\eta\in Y(\xi_{z_1}))(\nexists\eta'\in Y(\xi_{z_2})) \{CP(\eta')<CP(\eta)\}\}$.

The procedure for finding  the minimum value $C_{min}$ is outlined in Algorithm 1, which {proceeds} in three steps. First, each initial state $(x,N,y,N)\in q_0^{V'_{F}}$ gets a cost to be a pair of the form $(0,0)$; all other states in $Q^{V'_{F}}\backslash q_0^{V'_{F}}$ get a cost of the form $(\infty,\infty)$. 
Then, we run the following iteration for at most $4|X|^2\times(4|X|^2\times c_{max}+1)$ states in $Q_c$, where $4|X|^2$ is the maximum number of states of the modified verifier and $4|X|^2\times c_{max}+1$ is the maximum number of pairs of costs. For each state $(x, l_1, y, l_2)\in Q^{V'_{F}}$, there are at most $(|\Sigma|+1)(c_{max}+1)^2$ feasible transitions and $4|X|^2$ next states.
For a state $(x', l'_1, y', l'_2)$, the pair of costs is supposed to capture the minimal costs required to reach the pair of states $(x', l'_1)$ and $(y', l'_2)$. 
For each state $(x', l'_1, y', l'_2)$, we consider the following conditions: a) State $(x',l'_1,y',l'_2,c'_1,c'_2)$ is added to $Q^{V'_{MF}}$ if $(c'_1,c'_2)$ is minimal, i.e., we keep  $(c'_1,c'_2)$ if there exists a state $(x',l'_1,y',l'_2,c_1,c_2)$ such that ($c'_1\leq c_1$ and $c'_2< c_2$) or ($c'_1< c_1$ and $c'_2\leq c_2$), and eliminate $(x',l'_1,y',l'_2,c_1,c_2)$; b) If  $(c'_1,c'_2)$ is not minimal, we also keep  $(c'_1,c'_2)$ if it is incomparable, i.e., we keep  $(x',l'_1,y',l'_2,c'_1,c'_2)$ if there exists a state $(x',l'_1,y',l'_2,c_1,c_2)$ such that ($c'_1 > c_1$ and $c'_2 < c_2$) or ($c'_1 < c_1$ and $c'_2 > c_2$). Finally, we select the minimum cost of maximum value of left and right costs for all ending states with complexity $4|X|^2\times(4|X|^2\times c_{max}+1)$.
Overall, the total cost of finding $C_{min}$ would be $O(4|X|^2\times(4|X|^2\times c_{max}+1)\times (|\Sigma|+1)(c_{max}+1)^2 \times 4|X|^2\times (4|X|^2\times c_{max}+1)+4|X|^2\times(4|X|^2\times c_{max}+1))=O((4|X|^2)^4\times c^4_{max}\times |\Sigma|)$.

\begin{algorithm}[htbp] 
	\caption{Identification of minimum value $C_{min}$} 
	\label{A3} 
	\begin{algorithmic}[1] 
		\REQUIRE 
		A modified verifier $V'_{F}=AC(Q^{V'_{F}},((\Sigma\cup\{\varepsilon\})\times \mathbb{N})\times((\Sigma\cup\{\varepsilon\})\times \mathbb{N}), \delta_{V'_{F}}, q_{0}^{V'_{F}})$ and the set of ending states $X_e$.
		\ENSURE
		The minimum value $C_{min}$.
		\STATE $q_{0}^{V'_{MF}}:=\{(x,N,y,N,0,0)|(x,N,y,N)\in q_{0}^{V'_{F}}\}$;
		
		\STATE $Q':=\{(x,l_1,y,l_2,\infty,\infty)|(x,l_1,y,l_2)\in Q^{V'_{F}}\backslash q_{0}^{V'_{F}}\}$;
		
		\STATE $Q^{V'_{MF}}:=q_{0}^{V'_{MF}} \cup Q'$;
		$Q_c:=q_{0}^{V'_{MF}}$;
		$C_{min}:=\infty$;	
		
		\WHILE {$Q_c\neq \emptyset$}
		
		\FOR {each state $(x,l_1,y,l_2,c_x,c_y)\in Q_c$}
		
		\STATE $Q_c=Q_c \backslash \{(x,l_1,y,l_2,c_x,c_y)\}$;
		
		\FOR {each $\tau=((e,c),(e,c'))$ and\\ $\delta_{V'_{F}}((x,l_1,y,l_2),$ $\tau)$  $\neq$ $\emptyset $}\label{judgeEmpty}
		
		\FOR {each $(x',l'_1,y',l'_2)\in \delta_{V'_{F}}((x,l_1,y,l_2),$ $\tau)$}\label{newState}
		
		\STATE	$c'_1:= c_x+c$;	$c'_2 := c_y+c'$;
		
		UpdateCost$((x',l'_1,y',l'_2),(c'_1,c'_2))$;
		
		\ENDFOR 
		\ENDFOR
		\ENDFOR        
		\ENDWHILE
		
		\FOR {$(x,l_1,y,l_2)\in X_e$}
		\FOR {$(x,l_1,y,l_2,c_1,c_2)\in Q^{V'_{MF}}$}
		\IF{$\max(c_1,c_2)<C_{min}$}
		\STATE	$C_{min}=\max(c_1,c_2)$;
		\ENDIF
		\ENDFOR
		\ENDFOR
		\RETURN  $C_{min}$. 
		
		~~~~~~~~~	
		\STATE \textbf{procedure} {UpdateCost$((x',l'_1,y',l'_2),(c'_1,c'_2))$} 	
		
		\STATE $Q_d:=Q^{V'_{MF}}$;
		\FOR {each state $(x',l'_1,y',l'_2,c_1,c_2)\in Q_d$}		
		\IF {($c'_1 \leq c_1$ and $c'_2 < c_2$) or ($c'_1 < c_1$ and $c'_2 \leq c_2$)}
		
		\STATE	$Q^{V'_{MF}}=Q^{V'_{MF}}\cup \{(x',l'_1,y',l'_2,c'_1,c'_2)\}$ $\backslash$ $ \{(x'_i,l'_i,x'_j,l'_j,c_1,c_2)\}$;
		
		\STATE	$Q_c=Q_c\cup \{(x',l'_1,y',l'_2,c'_1,c'_2)\}$;
		
		\ELSE \IF {($c'_1 < c_1$ and $c'_2 > c_2$) or ($c'_1 > c_1$ and $c'_2 < c_2$)}
		
		\STATE $Q^{V'_{MF}}=Q^{V'_{MF}} \cup  \{(x',l'_1,y',l'_2,c'_1,c'_2)\}$;
		
		\STATE	$Q_c=Q_c\cup \{(x',l'_1,y',l'_2,c'_1,c'_2)\}$;
		
		\ENDIF
		\ENDIF
		\ENDFOR	
		\STATE \textbf{end procedure}
	\end{algorithmic}
\end{algorithm}

\begin{corollary}\label{extcordiag1}
	An attacker, under a constraint of a total cost $C_{min}$ on its tampering action,  has the capability to cause  a violation of $C'$-constrained tamper-tolerant diagnosability of $G_{nd}$ for arbitrarily large $C'$ ($C'\geq C_{min}$).
\end{corollary}
\begin{proof} 
	Suppose that the path that corresponds to $C_{min}$ is $\eta$: $\left\langle q^{V'_{F}}_1, \tau_1, q^{V'_{F}}_2,...,\tau_{n-1},q^{V'_{F}}_{n}\right\rangle$,
	where  $q^{V'_{F}}_1=(x_0,N,x'_0,N)\in q_{0}^{V'_{F}}$, $q^{V'_{F}}_{n}\in X_{ez}$, $q^{V'_{F}}_{i}=(x,l,x',l')$, $\tau_{i}=((e,c),(e,c'))$, $l=N$, $l'=N \vee F$, $i\in\{1,2,...,n-1\}$, and $CP(\eta)=C_{min}$. There exists $\sigma_f\in \Sigma_{f}$ such that $\tau_{j}=((\sigma_f,0),(\sigma_f,0))$ and  $j\in\{1,2,...,n-1\}$. 
	At state $q^{V'_{F}}_{j}=(x_j,l_j,x'_j,l'_j)$, $l_j=l'_j=N$. At state $q^{V'_{F}}_{j+1}=(x_{j+1},l_{j+1},x'_{j+1},l'_{j+1})$, $l_{j+1}=N$ and $l'_{j+1}=F$.
	There is a modified $F$-confused cycle defined after state $q^{V'_{F}}_{n}$, denoted by $\xi':\left\langle q^{V'_{F}}_{n}, \tau'_0, q^{V'_{F}}_{n+1},...,q^{V'_{F}}_{n+w}, \tau'_{w},q^{V'_{F}}_{n}\right\rangle$, where 
	for all $q^{V'_{F}}_{n+y}=(x,l,x',l')\in Q^{V'_{F}}$, $\tau'_{y}=((e',0),(e',0))$, $y
	\in\{0,1,2,...,w\}$, $l=N$, and $l'=F$. 
	
	If $C'\geq C_{min}$, the modified verifier $V'_F$ can construct modified $F$-confused cycles, i.e., at least the event $\sigma_f$ constructed above will be non-diagnosable  with a finite number of attacks (more generally, a total cost $C_{min}$).  
\end{proof}

\section{Conclusions}\label{Conclusion}

In this paper, we consider current-state estimation in a DES modeled as an NFA, under insertions, deletions, and substitutions {of observed symbols}. An observation automaton model is used to represent all possibly matching sequences {of observations}, which avoids explicitly enumerating all {such} sequences. {An algorithm is proposed} that is able to {systematically} {perform} this task. In order to ensure the property of {tamper-tolerant} diagnosability, a modified system is constructed, where attacks and costs are attached to the original plant. {Then, we verify the disgnosability of the plant under attacks through a verifier {with complexity that is polynomial in the size of the plant and the maximum value of the costs}.} A modified corrupted system and modified verifier are proposed to find the minimum value of $C$ that causes a violation of {tamper-tolerant} diagnosability for at least one fault.

In the future, we plan to develop ways to efficiently assess whether it is preferable to {perform} state estimation under multiple sensor measuring units.
We also plan to consider how state estimation {can be {achieved}} in the presence of {other} types of attackers. The cost of attacks is artificially assigned with respect to the likelihood of attack happening, which may {be challenging}. Hence, we plan to find an adaptive cost assignment function to dynamically adjust the likelihood of attacks.

\ifCLASSOPTIONcaptionsoff
  \newpage
\fi

\bibliographystyle{IEEEtran}
\bibliography{References}

\begin{thebibliography}{10}
\providecommand{\url}[1]{#1}
\csname url@samestyle\endcsname
\providecommand{\newblock}{\relax}
\providecommand{\bibinfo}[2]{#2}
\providecommand{\BIBentrySTDinterwordspacing}{\spaceskip=0pt\relax}
\providecommand{\BIBentryALTinterwordstretchfactor}{4}
\providecommand{\BIBentryALTinterwordspacing}{\spaceskip=\fontdimen2\font plus
\BIBentryALTinterwordstretchfactor\fontdimen3\font minus
  \fontdimen4\font\relax}
\providecommand{\BIBforeignlanguage}[2]{{%
\expandafter\ifx\csname l@#1\endcsname\relax
\typeout{** WARNING: IEEEtran.bst: No hyphenation pattern has been}%
\typeout{** loaded for the language `#1'. Using the pattern for}%
\typeout{** the default language instead.}%
\else
\language=\csname l@#1\endcsname
\fi
#2}}
\providecommand{\BIBdecl}{\relax}
\BIBdecl

\bibitem{simon2006optimal}
D.~Simon, \emph{Optimal State Estimation: Kalman, H Infinity, and Nonlinear
  Approaches}.\hskip 1em plus 0.5em minus 0.4em\relax New York: John Wiley \&
  Sons, 2006.

\bibitem{monticelli2012state}
A.~Monticelli, \emph{State Estimation in Electric Power Systems: A Generalized
  Approach}.\hskip 1em plus 0.5em minus 0.4em\relax Kluwer, Amsterdam: Springer
  Science \& Business Media, 2012.

\bibitem{zeigler2000theory}
B.~P. Zeigler, T.~G. Kim, and H.~Praehofer, \emph{Theory of Modeling and
  Simulation}.\hskip 1em plus 0.5em minus 0.4em\relax Academic Press, 2000.

\bibitem{ramadge1987supervisory}
P.~J. Ramadge and W.~M. Wonham, ``Supervisory control of a class of discrete
  event processes,'' \emph{SIAM Journal on Control and Optimization}, vol.~25,
  no.~1, pp. 206--230, 1987.

\bibitem{hadjicostis2020estimation}
C.~N. Hadjicostis, \emph{Estimation and Inference in Discrete Event
  Systems}.\hskip 1em plus 0.5em minus 0.4em\relax Springer Nature Switzerland
  AG, 2020.

\bibitem{sampath1995diagnosability}
M.~Sampath, R.~Sengupta, S.~Lafortune, K.~Sinnamohideen, and D.~Teneketzis,
  ``Diagnosability of discrete-event systems,'' \emph{IEEE Transactions on
  Automatic Control}, vol.~40, no.~9, pp. 1555--1575, 1995.

\bibitem{debouk2000coordinated}
R.~Debouk, S.~Lafortune, and D.~Teneketzis, ``Coordinated decentralized
  protocols for failure diagnosis of discrete event systems,'' \emph{Discrete
  Event Dynamic Systems}, vol.~10, no. 1-2, pp. 33--86, 2000.

\bibitem{debouk2003effect}
------, ``On the effect of communication delays in failure diagnosis of
  decentralized discrete event systems,'' \emph{Discrete Event Dynamic
  Systems}, vol.~13, no.~3, pp. 263--289, 2003.

\bibitem{bryans2008opacity}
J.~W. Bryans, M.~Koutny, L.~Mazar{\'e}, and P.~Y. Ryan, ``Opacity generalised
  to transition systems,'' \emph{International Journal of Information
  Security}, vol.~7, no.~6, pp. 421--435, 2008.

\bibitem{saboori2007notions}
A.~Saboori and C.~N. Hadjicostis, ``Notions of security and opacity in discrete
  event systems,'' in \emph{Proceedings of the 46th IEEE Conference on Decision
  and Control}.\hskip 1em plus 0.5em minus 0.4em\relax New Orleans, LA, USA:
  IEEE, 2007, Conference Proceedings, pp. 5056--5061.

\bibitem{jacob2016overview}
R.~Jacob, J.~J. Lesage, and J.~M. Faure, ``Overview of discrete event systems
  opacity: Models, validation, and quantification,'' \emph{Annual Reviews in
  Control}, vol.~41, pp. 135--146, 2016.

\bibitem{saboori2013verification}
A.~Saboori and C.~N. Hadjicostis, ``Verification of initial-state opacity in
  security applications of discrete event systems,'' \emph{Information
  Sciences}, vol. 246, pp. 115--132, 2013.

\bibitem{saboori2008opacity}
------, ``Opacity-enforcing supervisory strategies for secure discrete event
  systems,'' in \emph{Proceedings of the 47th IEEE Conference on Decision and
  Control}.\hskip 1em plus 0.5em minus 0.4em\relax Cancun, Mexico: IEEE, 2008,
  Conference Proceedings, pp. 889--894.

\bibitem{wakaiki2017supervisory}
M.~Wakaiki, P.~Tabuada, and J.~P. Hespanha, ``Supervisory control of
  discrete-event systems under attacks,'' \emph{Dynamic Games and
  Applications}, pp. 965--983, 2019.

\bibitem{hu2018state}
L.~Hu, Z.~D. Wang, Q.~L. Han, and X.~H. Liu, ``State estimation under false
  data injection attacks: Security analysis and system protection,''
  \emph{Automatica}, vol.~87, pp. 176--183, 2018.

\bibitem{athanasopoulou2010maximum}
E.~Athanasopoulou, L.~X. Li, and C.~N. Hadjicostis, ``Maximum likelihood
  failure diagnosis in finite state machines under unreliable observations,''
  \emph{IEEE Transactions on Automatic Control}, vol.~55, no.~3, pp. 579--593,
  2010.

\bibitem{carvalho2011generalized}
L.~K. Carvalho, M.~V. Moreira, and J.~C. Basilio, ``Generalized robust
  diagnosability of discrete event systems,'' in \emph{Proceedings of the 18th
  IFAC World Congress}, Milano, Italy, 2011, Conference Proceedings, pp.
  8737--8742.

\bibitem{carvalho2012robust}
L.~K. Carvalho, J.~C. Basilio, and M.~V. Moreira, ``Robust diagnosis of
  discrete event systems against intermittent loss of observations,''
  \emph{Automatica}, vol.~48, no.~9, pp. 2068--2078, 2012.

\bibitem{carvalho2013robust}
L.~K. Carvalho, M.~V. Moreira, J.~C. Basilio, and S.~Lafortune, ``Robust
  diagnosis of discrete-event systems against permanent loss of observations,''
  \emph{Automatica}, vol.~49, no.~1, pp. 223--231, 2013.

\bibitem{mousavinejad2018novel}
E.~Mousavinejad, F.~W. Yang, Q.~L. Han, and L.~Vlacic, ``A novel cyber attack
  detection method in networked control systems,'' \emph{IEEE Transactions on
  Cybernetics}, vol.~48, no.~11, pp. 3254--3264, 2018.

\bibitem{ding2018survey}
D.~Ding, Q.~L. Han, Y.~Xiang, X.~H. Ge, and X.~M. Zhang, ``A survey on security
  control and attack detection for industrial cyber-physical systems,''
  \emph{Neurocomputing}, vol. 275, pp. 1674--1683, 2018.

\bibitem{li2013robust}
X.~Li and A.~Scaglione, ``Robust decentralized state estimation and tracking
  for power systems via network gossiping,'' \emph{IEEE Journal on Selected
  Areas in Communications}, vol.~31, no.~7, pp. 1184--1194, 2013.

\bibitem{zhao2015power}
J.~B. Zhao, G.~X. Zhang, K.~Das, G.~N. Korres, N.~M. Manousakis, A.~K. Sinha,
  and Z.~Y. He, ``Power system real-time monitoring by using {PMU}-based robust
  state estimation method,'' \emph{IEEE Transactions on Smart Grid}, vol.~7,
  no.~1, pp. 300--309, 2015.

\bibitem{farwell2011stuxnet}
J.~P. Farwell and R.~Rohozinski, ``Stuxnet and the future of cyber war,''
  \emph{Survival}, vol.~53, no.~1, pp. 23--40, 2011.

\bibitem{slay2007lessons}
J.~Slay and M.~Miller, ``Lessons learned from the maroochy water breach,'' in
  \emph{Proceedings of International Conference on Critical Infrastructure
  Protection}, vol. 253.\hskip 1em plus 0.5em minus 0.4em\relax Boston, MA,
  USA: Springer, 2007, Conference Proceedings, pp. 73--82.

\bibitem{kerns2014unmanned}
A.~J. Kerns, D.~P. Shepard, J.~A. Bhatti, and T.~E. Humphreys, ``Unmanned
  aircraft capture and control via gps spoofing,'' \emph{Journal of Field
  Robotics}, vol.~31, no.~4, pp. 617--636, 2014.

\bibitem{goes2017stealthy}
R.~M. G{\'o}es, E.~Kang, R.~Kwong, and S.~Lafortune, ``Stealthy deception
  attacks for cyber-physical systems,'' in \emph{Proceedings of the 56th IEEE
  Conference on Decision and Control}.\hskip 1em plus 0.5em minus 0.4em\relax
  Melbourne, VIC, Australia: IEEE, 2017, Conference Proceedings, pp.
  4224--4230.

\bibitem{meira2019synthesis}
R.~Meira-G{\'o}es, R.~Kwong, and S.~Lafortune, ``Synthesis of sensor deception
  attacks for systems modeled as probabilistic automata,'' in \emph{Proceedings
  of American Control Conference}.\hskip 1em plus 0.5em minus 0.4em\relax
  Philadelphia, PA, USA: IEEE, 2019, Conference Proceedings, pp. 5620--5626.

\bibitem{su2017cyber}
R.~Su, ``A cyber attack model with bounded sensor reading alterations,'' in
  \emph{Proceedings of American Control Conference}.\hskip 1em plus 0.5em minus
  0.4em\relax Seattle, WA, USA: IEEE, 2017, Conference Proceedings, pp.
  3200--3205.

\bibitem{su2018supervisor}
------, ``Supervisor synthesis to thwart cyber attack with bounded sensor
  reading alterations,'' \emph{Automatica}, vol.~94, pp. 35--44, 2018.

\bibitem{thorsley2008diagnosability}
D.~Thorsley, T.~S. Yoo, and H.~E. Garcia, ``Diagnosability of stochastic
  discrete-event systems under unreliable observations,'' in \emph{Proceedings
  of American Control Conference}.\hskip 1em plus 0.5em minus 0.4em\relax
  Seattle, WA, USA: IEEE, 2008, Conference Proceedings, pp. 1158--1165.

\bibitem{athanasopoulou2006probabilistic}
E.~Athanasopoulou, L.~X. Li, and C.~N. Hadjicostis, ``Probabilistic failure
  diagnosis in finite state machines under unreliable observations,'' in
  \emph{Proceedings of the 8th International Workshop on Discrete Event
  Systems}.\hskip 1em plus 0.5em minus 0.4em\relax Ann Arbor, MI, USA: IEEE,
  2006, Conference Proceedings, pp. 301--306.

\bibitem{boel2002decentralized}
R.~K. Boel and J.~H. Van~Schuppen, ``Decentralized failure diagnosis for
  discrete-event systems with costly communication between diagnosers,'' in
  \emph{Proceedings of the 6th International Workshop on Discrete Event
  Systems}, Zaragoza, Spain, 2002, Conference Proceedings, pp. 175--181.

\bibitem{khanna1973sampling}
M.~Khanna, ``Sampling and transmission policies for controlled {Markov}
  processes with costly communication,'' PhD thesis, Department of Electrical
  Engineering, University of Toronto, Toronto, 1973.

\bibitem{lin2014control}
F.~Lin, ``Control of networked discrete event systems: Dealing with
  communication delays and losses,'' \emph{SIAM Journal on Control and
  Optimization}, vol.~52, no.~2, pp. 1276--1298, 2014.

\bibitem{lin2019state}
F.~Lin, W.~Wang, L.~T. Han, and B.~Shen, ``State estimation of multi-channel
  networked discrete event systems,'' \emph{IEEE Transactions on Control of
  Network Systems}, vol.~7, no.~1, pp. 53--63, 2019.

\bibitem{rudie2003minimal}
K.~Rudie, S.~Lafortune, and F.~Lin, ``Minimal communication in a distributed
  discrete-event system,'' \emph{IEEE Transactions on Automatic Control},
  vol.~48, no.~6, pp. 957--975, 2003.

\bibitem{householder2001managing}
A.~Householder, A.~Manion, L.~Pesante, G.~M. Weaver, and R.~Thomas, ``Managing
  the threat of denial-of-service attacks,'' Carnegie-Mellon Univ Software
  Engineering Inst, Pittsburgh, PA, USA, Report, 2001.

\bibitem{cassandras2009introduction}
C.~G. Cassandras and S.~Lafortune, \emph{Introduction to Discrete Event
  Systems}.\hskip 1em plus 0.5em minus 0.4em\relax New York: Springer Science
  \& Business Media, 2009.

\bibitem{yoo2002polynomial}
T.~S. Yoo and S.~Lafortune, ``Polynomial-time verification of diagnosability of
  partially observed discrete-event systems,'' \emph{IEEE Transactions on
  Automatic Control}, vol.~47, no.~9, pp. 1491--1495, 2002.

\end{thebibliography}

\end{document}